\documentclass[aps,pra,preprint,amsmath,amssymb,superscriptaddress,nofootinbib]{revtex4-2}

\usepackage[utf8]{inputenc}
\usepackage{bm}
\usepackage{mathtools}
\usepackage{amsthm}
\usepackage{booktabs}
\usepackage{graphicx}
\usepackage{tikz}
\usetikzlibrary{arrows.meta,patterns}
\usepackage{hyperref}

\newcommand{\Tr}{\operatorname{Tr}}
\newcommand{\Id}{\mathbb{1}}
\newcommand{\Om}{\Omega}
\newcommand{\grad}{\operatorname{grad}}
\newcommand{\M}{\mathcal M}
\DeclareMathOperator{\arccosh}{arccosh}

\newtheorem{theorem}{Theorem}
\newtheorem{proposition}[theorem]{Proposition}

\newtheorem{corollary}[theorem]{Corollary}
\theoremstyle{definition}
\newtheorem{definition}[theorem]{Definition}
\theoremstyle{remark}
\newtheorem{remark}[theorem]{Remark}

\begin{document}

\title{The Global Geometry of the Gaussian Bures Manifold:\\
Admissible Domain, Spectral Boundary, and Asymptotic Classicality}

\author{Christian Kerskens}
\affiliation{Trinity College Institute of Neuroscience, Trinity College Dublin, Dublin, Ireland}
\date{\today}

\begin{abstract}
We construct the covariance-sector Bures geometry of centered bosonic Gaussian
states from the Gaussian symmetric-logarithmic-derivative equation
and determine its global quantum domain.  The Robertson--Schr\"odinger
condition bounds this domain below by the Williamson uncertainty floor.  Its
boundary geometry is anisotropic: spectrum-changing radial coefficients
diverge, whereas squeezing and rotation coefficients remain finite along the
pure Gaussian orbit.  Dually, the covariance cometric develops an exact
\(m^2\)-dimensional kernel when \(m\) Williamson modes are pure, although the
floor remains at finite radial Bures distance.  We derive the radial potential
\(\Phi=-\tfrac12\sum_k\log(\nu_k^2-\tfrac14)\), which generates covariance
dilation on fixed-Williamson-frame radial submanifolds and equals \(\beta F\)
on fixed-Hamiltonian thermal families.  At large symplectic eigenvalue the
relative symplectic correction is suppressed as \(O(\hbar^2/\nu^2)\), yielding
an intrinsic high-noise recovery of covariance Fisher--Rao geometry.  This
limit is distinct from the Williamson floor, a quantum pure-state boundary
rather than a classical limit: classical statistics there require an
externally supplied measurement channel, which the geometry constrains but
does not select.  As a secondary comparison with the Fisher--Rao and
Bures--Wasserstein geometries, and after introducing a constant
action-matching parameter \(\eta\), the three normalized determinant densities
possess an exact one-mode junction at
\((x,\eta)=(\sqrt{\varphi},\sqrt{\varphi})\).  The equality is kinematic:
neither an additional boundary nor a dynamical transition, and branch selection
requires an explicitly stated minimum-log-determinant hypothesis.  Finally,
along any bounded-rate compression toward the floor the Bures kinetic action
diverges while the Bures--Wasserstein cost stays finite; under a bounded Bures
budget the consequence is radial deceleration, not a forced change of geometry.
\end{abstract}

\maketitle

\section{Introduction}
\label{sec:intro}

The manifolds of classical and quantum Gaussian states provide a setting in
which statistical, transport, and symplectic structures can be compared in
closed form.  Centered classical Gaussian measures are parametrized by positive
covariance matrices and carry both the covariance-restricted Fisher--Rao metric
and the Bures--Wasserstein transport metric
\cite{Villani2009,Amari2016,Takatsu2011,Bhatia2019,ThanwerdasPennec2023,WongYang2022}.  Centered bosonic Gaussian quantum
states occupy a proper subdomain of the same positive cone and carry the Bures
metric induced by quantum fidelity \cite{Weedbrook2012,BraunsteinCaves1994}.
These geometries should not be identified merely because they share related
matrix formulas: their operational meanings and their domains differ.

Previous work identified the algebraic origin of the Gaussian Bures correction
as a symplectic phase-bundle reduction \cite{KerskensLMP}.  The present paper
derives the required dual covariance formula directly from the standard
Gaussian SLD equation and addresses the complementary global question: on what
domain is this metric positive and nondegenerate, what geometric structure
survives at its natural boundary, and how does the quantum correction behave in
the opposite, high-noise regime?

The intrinsic answer has three parts.  First, quantum covariance matrices
satisfy the Robertson--Schr\"odinger condition
\cite{SimonMukundaDutta1994}
\begin{equation}
    \Sigma+\frac{i\hbar}{2}\Om\geq0,
    \label{eq:uncertainty-intro}
\end{equation}
or equivalently \(\nu_k\geq\hbar/2\) for every symplectic eigenvalue.  This
condition removes a neighborhood of the origin from the classical positive
covariance cone.  Classical covariance dilation can therefore be continued
toward the origin indefinitely, whereas its restriction to the quantum domain
reaches the uncertainty boundary after finite flow parameter.

Second, the Bures geometry becomes anisotropic at this boundary.  For one mode,
the spectrum-changing radial coefficient diverges on approach from the
full-rank interior, while squeezing and phase-rotation coefficients remain
finite.  The pure Gaussian states therefore retain a regular intrinsic
geometry along their unitary orbit even though the mixed-state radial chart
becomes singular.  In the dual description the same fact appears as a loss of
rank.  The Williamson-frame decomposition extends this statement to arbitrary
mode number: if exactly \(m\) symplectic eigenvalues reach the floor, the dual
covariance form has nullity \(m^2\).  Despite this radial singularity, the
Williamson floor lies at finite Bures distance.  The associated scaling defect
\(\Phi=-\tfrac12\sum_k\log(\nu_k^2-\tfrac14)\) is a Bures potential for
covariance dilation on fixed-Williamson-frame radial submanifolds and reduces
to \(\beta F\) on fixed-Hamiltonian thermal families.

Third, there is a distinct classical regime at large symplectic eigenvalue.  In
the exact dual formula the Fisher-type term grows quadratically with covariance,
whereas the symplectic correction remains of order \(\hbar^2\).  Its relative
weight therefore vanishes as \(\hbar^2/\nu^2\).  This limit recovers the
covariance Fisher geometry to leading order.  It does not, by itself, establish
an interpolation to Bures--Wasserstein transport; any such transport reduction
requires a separately specified map or dynamics.

The two regimes have different logical status.  The high-noise limit is an
intrinsic asymptotic recovery of the covariance Fisher geometry.  The lower
spectral edge is not a classical limit: it is a pure-state quantum boundary
with finite tangential geometry and singular radial behavior.  Classical
statistics can arise there only after an externally supplied commutative or
conditioned restriction.  A specified measurement channel then maps the
quantum boundary family to classical outcome distributions and determines
which geometric directions are statistically accessible.  Measurement thus
provides an operational interface between the quantum boundary geometry and
classical statistics.  The covariance metric constrains the geometry on which
this restriction acts, but it does not select the channel or derive collapse
dynamics, a stochastic record, or a measurement apparatus.

Beyond these intrinsic results, we compare the Gaussian Bures covariance
geometry with two classical covariance geometries: the Fisher--Rao information
metric and the Bures--Wasserstein transport metric.  Introducing a constant
action-matching parameter \(\eta\), we show that the corresponding normalized
determinant densities possess an exact one-mode junction.  This equality
organizes the possible branch diagrams of the three geometries.  It is,
however, purely kinematic: it is neither an additional boundary of the Gaussian
Bures manifold nor, by itself, a dynamical transition.  Any branch-selection
interpretation requires the explicitly stated minimum-log-determinant
hypothesis.

We further compare the primal kinetic costs assigned by the Bures and
Bures--Wasserstein metrics to the same prescribed radial trajectory.  Under a
finite-rate compression toward the Williamson floor, the Bures kinetic density
and integrated action diverge, whereas the Bures--Wasserstein cost remains
finite.  This comparison establishes a conditional finite-rate cost separation
rather than a mechanism forcing a transition between the two geometries.  Under
a bounded Bures kinetic budget, the intrinsic consequence is radial
deceleration.

\paragraph{Hierarchy of claims.}
The manuscript keeps five levels logically separate.  The admissible domain,
boundary spectrum, radial potential, and high-noise asymptotics are intrinsic
results about the Gaussian Bures geometry.  Measurement-induced boundary
statistics require an explicitly supplied operational restriction.  The
three-volume equality is an exact kinematic comparison after an action scale
has been introduced for the transport form.  The branch topology is
conditional on a stated determinant selection rule.  A dynamical or
thermodynamic realization of that rule remains open.  No conclusion at a later
level is used to prove a result at an earlier one.

The paper is organized as follows.  Section~\ref{sec:manifold} constructs the
Gaussian Bures covariance geometry and determines its quantum-admissible
domain.  Section~\ref{sec:radial} develops the radial, tangential, and dual
descriptions of the Williamson boundary.  Section~\ref{sec:potential} derives
the radial dilation potential, and Section~\ref{sec:upper} establishes the
intrinsic high-noise Fisher--Rao limit.  Section~\ref{sec:triple-volume}
presents the auxiliary comparison with Fisher--Rao and Bures--Wasserstein
geometry, including the exact one-mode junction, its conditional branch
topology, and the finite-rate radial cost separation.  Section
\ref{sec:two-routes} assembles the global interpretation, distinguishes
intrinsic high-noise classicality from measurement-induced boundary statistics,
distinguishes Gaussian boundary motion from model-dependent non-Gaussian
escape, and states the limitations and open directions.

\section{Gaussian Bures geometry and its quantum domain}
\label{sec:manifold}

\subsection{Bures construction and conventions}

Let \(\hat R=(\hat x_1,\hat p_1,\ldots,\hat x_N,\hat p_N)^{\mathsf T}\) be the
quadrature vector, with
\begin{equation}
    [\hat R_a,\hat R_b]=i\hbar\Om_{ab},
    \qquad
    \Om=\bigoplus_{k=1}^N
    \begin{pmatrix}0&1\\-1&0\end{pmatrix}.
\end{equation}
For a centered state, the covariance matrix is
\begin{equation}
    \Sigma_{ab}=\frac12\langle\{\hat R_a,\hat R_b\}\rangle.
\end{equation}
We use the convention in which the vacuum covariance is
\(\Sigma_{\mathrm{vac}}=(\hbar/2)\Id\).  Most explicit formulas below set
\(\hbar=1\), so the Williamson floor is \(\nu=1/2\).  Drafts using a floor at
\(\nu=1\) are related by \(\nu_{(1)}=2\nu_{(1/2)}\); all coefficients must be
transformed consistently when comparing conventions.

Williamson's theorem gives \cite{Williamson1936,deGosson2006,Weedbrook2012}
\begin{equation}
    \Sigma=S_W\left(\bigoplus_{k=1}^N\nu_k\Id_2\right)S_W^{\mathsf T},
    \qquad S_W\Om S_W^{\mathsf T}=\Om,
    \label{eq:williamson}
\end{equation}
where \(\nu_k\geq\hbar/2\).  We write
\begin{align}
    \M_Q
    &=\left\{\Sigma\in\operatorname{Sym}_{++}(2N):
      \Sigma+\frac{i\hbar}{2}\Om\geq0\right\},\\
    \M_Q^{\circ}
    &=\left\{\Sigma\in\M_Q:\nu_k>\frac{\hbar}{2}\ \text{for all }k\right\}.
\end{align}
The Bures metric, originating from the transition probability between mixed states \cite{Bures1969,Uhlmann1976} and established as the minimal monotone quantum metric \cite{Petz1996}, is smooth on the faithful-state interior \(\M_Q^{\circ}\). Closed Gaussian-state fidelity formulas underlying this metric are given in Ref.~\cite{Banchi2015}.
At rank-changing boundaries, the limiting interior QFI must be distinguished
from the metric intrinsic to a fixed-rank boundary stratum
\cite{Safranek2017boundary}.

Following the standard formalism of quantum estimation theory \cite{Helstrom1976,Holevo2011,Paris2009}, for a smooth state family \(\rho_\theta\), the symmetric logarithmic derivative
\(L_i\) and QFI are
\begin{equation}
    \partial_i\rho=\frac12(L_i\rho+\rho L_i),
    \qquad
    F_{ij}=\frac12\Tr\!\left[\rho(L_iL_j+L_jL_i)\right].
\end{equation}
Our Bures convention is
\begin{equation}
    g^{\rm B}_{ij}=\frac14F_{ij}.
    \label{eq:bures-convention}
\end{equation}
For covariance-only variations, a real symmetric gain matrix
\(\mathfrak G_i\) satisfies
\begin{equation}
    \Sigma\mathfrak G_i\Sigma
    +\frac{\hbar^2}{4}\Om\mathfrak G_i\Om
    =\partial_i\Sigma,
    \qquad
    F_{ij}=\frac12\Tr(\mathfrak G_i\partial_j\Sigma).
    \label{eq:lyapunov}
\end{equation}
Equivalent general formulas are given in Refs.~\cite{Monras2013,Safranek2018};
information matrices and symmetric logarithmic derivatives for bosonic Gaussian
thermal states parameterized by their Hamiltonian matrices are derived in
Ref.~\cite{HuangWilde2024}.

To invert this metric, define the self-adjoint superoperator
\begin{equation*}
    \mathcal S_\Sigma(G)
    =\Sigma G\Sigma+\frac{\hbar^2}{4}\Om G\Om.
\end{equation*}
For symmetric tangent vectors \(X,Y\), Eq.~\eqref{eq:lyapunov} and the
convention \(g^{\rm B}=F/4\) give
\begin{equation*}
    g_{\rm B}(X,Y)
    =\frac18\Tr\!\left[\mathcal S_\Sigma^{-1}(X)Y\right].
\end{equation*}
Thus, with the trace pairing \(\langle P,X\rangle=\Tr(PX)\), the cometric
operator is \(8\mathcal S_\Sigma\), and
\begin{equation}
    g_{\rm B}^*(\Sigma;P,P)
    =8\Tr(P\Sigma P\Sigma)
    +2\hbar^2\Tr(P\Om P\Om),
    \label{eq:dual-schematic}
\end{equation}
for symmetric covectors \(P\).  The plus sign in front of
\(\Tr(P\Om P\Om)\) is essential; this trace is negative on the radial scalar
sector because \(\Om P\Om=-P\) there.  In units \(\hbar=1\),
Eq.~\eqref{eq:dual-schematic} becomes
\begin{equation}
    g_{\rm B}^*(\Sigma;P,P)
    =8\Tr(P\Sigma P\Sigma)+2\Tr(P\Om P\Om),
    \label{eq:dual-explicit}
\end{equation}
which is the convention used throughout the remainder of the paper.  The
coefficient pair, the resulting one-mode spectra, and every determinant
identity, class weight, and branch ordering quoted below are verified by an
archived, deterministic symbolic and numerical script~\cite{VerifyScript}.
Within the two-parameter ansatz
\begin{equation}
    A\Tr(P\Sigma P\Sigma)+B\hbar^2\Tr(P\Om P\Om),
\end{equation}
the coefficient pair is uniquely fixed as \((A,B)=(8,2)\).  On the radial
sector at
\(\Sigma=\nu\Id_2\), the covector \(P=p\Id_2\) has coordinate component
\(P_\nu=\Tr(P\,\partial_\nu\Sigma)=2p\), and Eq.~\eqref{eq:dual-explicit}
reproduces \(g_{\rm B}^{\nu\nu}P_\nu^2=16(\nu^2-\tfrac14)p^2\); the squeezing
sector likewise returns \(16(\nu^2+\tfrac14)p^2\).  These two independent
sectors determine both coefficients.

Although the algebraic origin of the symplectic correction is analyzed in the
companion construction \cite{KerskensLMP}, the metric results below are
logically self-contained: Eq.~\eqref{eq:dual-schematic} follows here from the
Gaussian SLD equation and the trace pairing, and every subsequent argument uses
only this explicit form.

\subsection{The quantum-admissible domain}
\label{sec:domain}

The positive covariance cone of centered classical Gaussians is invariant under
uniform dilation
\begin{equation}
    D_s:\Sigma\longmapsto e^s\Sigma,
    \qquad s\in\mathbb R,
    \label{eq:dilation}
\end{equation}
generated by the Euler vector field \(V_\Sigma=\Sigma\).  On the classical cone,
this flow is defined for every finite \(s\).  The quantum domain is not invariant
under backward dilation because \(\Om\) is fixed.

\begin{proposition}[Finite-parameter contact with the quantum boundary]
\label{prop:finite-contact}
Let \(\Sigma_0\in\M_Q^{\circ}\) have smallest symplectic eigenvalue
\(\nu_{\min}(\Sigma_0)\).  The backward dilation orbit
\(\Sigma(s)=e^s\Sigma_0\) remains quantum admissible precisely when
\begin{equation}
    e^s\nu_{\min}(\Sigma_0)\geq\frac{\hbar}{2}.
\end{equation}
It first meets the boundary at the finite parameter
\begin{equation}
    s_*=\log\frac{\hbar}{2\nu_{\min}(\Sigma_0)}<0.
\end{equation}
\end{proposition}

\begin{proof}
Uniform covariance dilation multiplies every symplectic eigenvalue by \(e^s\).
The uncertainty condition is therefore equivalent to
\(e^s\nu_k(\Sigma_0)\geq\hbar/2\) for all \(k\).  The first saturated inequality
is the one associated with \(\nu_{\min}\), giving the stated value of \(s_*\).
\end{proof}

\begin{remark}
Proposition~\ref{prop:finite-contact} concerns completeness of a particular
restricted vector field, not geodesic completeness of the Bures metric.  In
fact, Section~\ref{sec:radial} shows that the one-mode boundary is at finite
radial Bures distance even though the radial metric coefficient diverges.
\end{remark}

The boundary is stratified by the set of Williamson modes that reach the
floor.  The exact covariance-cometric nullity of each such stratum is derived
in Theorem~\ref{thm:multimode-spectrum}.

This intrinsic Williamson boundary should not be confused with an externally
imposed ordinary eigenvalue floor in classical covariance mechanics.  In the
Hamiltonian lift of Bures--Wasserstein covariance dynamics developed in
Ref.~\cite{KerskensBWHamiltonian}, a condition of the form
\(\Sigma\succ\sigma_{\rm f}\Id\) is enforced by an added logarithmic barrier.
The resulting divergence is a potential-induced stiffness.  Here the boundary
is instead fixed by quantum admissibility, and the singularity is an
anisotropic rank change of the quantum Bures cometric itself.

\section{Geometry of the Williamson boundary}
\label{sec:radial}

\subsection{Radial singularity and finite Bures distance}

For fixed Williamson frame \(S_W\), define the radial submanifold
\begin{equation}
    \M_{S_W}
    =\left\{
    S_W\left(\bigoplus_{k=1}^{N}\nu_k\Id_2\right)S_W^{\mathsf T}:
    \nu_k>\frac12
    \right\},
    \label{eq:radial-manifold}
\end{equation}
where this and the following explicit formulas set \(\hbar=1\).  The Euler field
restricts to
\begin{equation}
    V=\sum_{k=1}^{N}\nu_k\partial_{\nu_k}.
    \label{eq:euler}
\end{equation}
The fixed-frame restriction is essential: varying \(S_W\) introduces
mode-mixing, entangling, squeezing, and rotation directions.

For one mode we write
\begin{equation}
    \Sigma(\nu,r,\varphi)=\nu G(r,\varphi),
    \qquad
    G=R(\varphi)
    \begin{pmatrix}e^{2r}&0\\0&e^{-2r}\end{pmatrix}
    R(\varphi)^{\mathsf T},
    \label{eq:onemode-covariance}
\end{equation}
with
\begin{equation}
    R(\varphi)=
    \begin{pmatrix}
    \cos\varphi&-\sin\varphi\\
    \sin\varphi&\cos\varphi
    \end{pmatrix}.
\end{equation}
Here \(V=\nu\partial_\nu\) at fixed \((r,\varphi)\).

\begin{proposition}[Single-mode Gaussian QFI]
\label{prop:single-mode-qfi}
For the covariance family in Eq.~\eqref{eq:onemode-covariance}, with
\(\nu>1/2\), the QFI matrix in coordinates \((\nu,r,\varphi)\) is diagonal
(verified in Ref.~\cite{VerifyScript}, check~C1):
\begin{equation}
    F_{ij}=\operatorname{diag}\!\left(
    \frac{1}{\nu^2-\frac14},
    \frac{4\nu^2}{\nu^2+\frac14},
    \frac{4\nu^2\sinh^2(2r)}{\nu^2+\frac14}
    \right).
    \label{eq:full-qfi}
\end{equation}
The angle \(\varphi\) is the phase-space rotation angle.  If instead the
squeeze-operator phase is \(\theta=2\varphi\), then
\(F_{\theta\theta}=F_{\varphi\varphi}/4\).
\end{proposition}

The derivation from Eq.~\eqref{eq:lyapunov} is given in
Appendix~\ref{app:qfi}; after conversion of covariance and angle conventions,
the result agrees with the general single-mode expressions of
Ref.~\cite{Pinel2013}.  The radial entry can also be checked in the thermal
Fock basis.  With mean occupation \(\bar n=\nu-1/2\),
\begin{equation}
    p_m=\frac{\bar n^m}{(1+\bar n)^{m+1}},
\end{equation}
and the classical Fisher information of these commuting probabilities is
\begin{equation}
    \sum_m p_m(\partial_\nu\log p_m)^2
    =\frac{1}{\bar n(1+\bar n)}
    =\frac{1}{\nu^2-\frac14}.
\end{equation}

The radial Bures coefficient and its inverse are
\begin{equation}
    g^{\rm B}_{\nu\nu}
    =\frac{1}{4(\nu^2-\frac14)},
    \qquad
    g_{\rm B}^{\nu\nu}=4(\nu^2-\tfrac14).
    \label{eq:radial-components}
\end{equation}
Thus the primal radial coefficient diverges, while the dual radial mobility
vanishes, as \(\nu\downarrow1/2\).

\begin{proposition}[Finite radial distance to the pure boundary]
\label{prop:finite-distance}
For \(\nu_0>1/2\), the radial Bures length at fixed \((r,\varphi)\) from
\(\nu_0\) to the boundary is finite:
\begin{equation}
    \ell(\nu_0,\tfrac12)
    =\int_{1/2}^{\nu_0}\sqrt{g^{\rm B}_{\nu\nu}}\,d\nu
    =\frac12\arccosh(2\nu_0).
    \label{eq:finite-distance}
\end{equation}
\end{proposition}

The combination of a divergent coordinate coefficient and finite integrated
distance is not contradictory.  Near the boundary, writing
\(\delta=\nu-1/2\), one has
\(g^{\rm B}_{\nu\nu}\sim(4\delta)^{-1}\), whose square root is integrable.

\subsection{Tangential geometry at the pure Gaussian boundary}
\label{sec:tangential}

Equation~\eqref{eq:full-qfi} separates the spectrum-changing direction from
unitary shape changes.  Taking the interior limit gives
\begin{equation}
    F_{rr}\longrightarrow2,
    \qquad
    F_{\varphi\varphi}\longrightarrow2\sinh^2(2r),
    \qquad \nu\downarrow\frac12.
    \label{eq:tangential-limits}
\end{equation}
The coordinate \(\varphi\) becomes redundant at \(r=0\), explaining the
vanishing angular coefficient there.

\begin{theorem}[Anisotropic pure-state boundary]
\label{thm:anisotropic-boundary}
On the one-mode Gaussian manifold, the mixed-state interior Bures metric is
singular only in the spectrum-changing radial direction as
\(\nu\downarrow1/2\).  The squeezing and phase-rotation coefficients have finite
limits and define the intrinsic Bures, equivalently Fubini--Study, geometry
tangent to the pure Gaussian-state manifold.
\end{theorem}

\begin{proof}
The divergence of the radial component and the finite limits of the two shape
components follow from Proposition~\ref{prop:single-mode-qfi} and
Eq.~\eqref{eq:bures-convention}.  At \(\nu=1/2\), variations in \(r\) and
\(\varphi\) are generated by Gaussian unitaries and remain within the pure-state
orbit.  The Bures metric restricted to pure states equals the Fubini--Study
metric, with the usual convention-dependent constant
\cite{BengtssonZyczkowski2017}.
\end{proof}

Theorem~\ref{thm:anisotropic-boundary} rules out a simple identification of the
spectral floor with the end of Gaussian geometry.  What ends is the faithful
radial direction.  Pure Gaussian motion tangent to the boundary remains
available.

\subsection{Degeneration of the covariance cometric}
\label{sec:dual-degeneration}

The radial divergence has an equivalent cotangent description.

\begin{theorem}[Radial nullity at the pure-state boundary]
\label{thm:null-mobility}
Consider a one-mode centered Gaussian state at its Williamson representative
\(\Sigma=\nu\Id_2\), and define
\begin{equation}
    x=\frac{2\nu}{\hbar}.
\end{equation}
Let
\begin{equation}
    E_0=\frac{\Id_2}{\sqrt2}
\end{equation}
be the trace-normalized radial covector, and let \(E_1,E_2\) be any
trace-orthonormal basis of the traceless symmetric covectors.  Then
Eq.~\eqref{eq:dual-schematic} gives
\begin{align}
    g_{\rm B}^*(E_0,E_0)
    &=2\hbar^2(x^2-1),\label{eq:radial-dual-eigenvalue}\\
    g_{\rm B}^*(E_a,E_a)
    &=2\hbar^2(x^2+1),
    \qquad a=1,2,\label{eq:shape-dual-eigenvalues}
\end{align}
and all cross terms vanish at the Williamson representative.

Consequently, as \(x\downarrow1\), the radial eigenvalue vanishes while the
two shape eigenvalues approach \(4\hbar^2\).  The limiting one-mode covariance
cometric therefore loses exactly one rank.  Equivalently, the primal metric
diverges in the spectrum-changing radial direction while the intrinsic metric
tangent to the pure Gaussian orbit remains finite.

The Bures covariance geometry consequently has no positive-definite radial
continuation through the Williamson floor.  This is a statement about the
geometry of the admissible covariance domain; it does not by itself imply
dynamical arrest or specify motion along the boundary. The two eigenvalue formulas and the vanishing cross terms are verified in Ref.~\cite{VerifyScript} (check~C2).
\end{theorem}

\begin{proof}
At \(\Sigma=\nu\Id_2\), the scalar covector satisfies
\(\Om E_0\Om=-E_0\), whereas each traceless symmetric covector satisfies
\(\Om E_a\Om=E_a\).  Since \(\Tr(E_\alpha^2)=1\), substitution into
Eq.~\eqref{eq:dual-schematic}, followed by \(\nu=x\hbar/2\), yields
Eqs.~\eqref{eq:radial-dual-eigenvalue} and
\eqref{eq:shape-dual-eigenvalues}.  Orthogonality of the scalar and traceless
sectors gives the vanishing cross terms.  The limiting rank and the corresponding
primal behavior follow by inversion on the faithful interior.
\end{proof}

\begin{theorem}[Multimode Williamson-frame spectrum and boundary nullity]
\label{thm:multimode-spectrum}
Let
\begin{equation}
    \Sigma=SDS^{\mathsf T},
    \qquad
    D=\bigoplus_{k=1}^{N}\nu_k\Id_2,
    \qquad
    S\in\operatorname{Sp}(2N,\mathbb R),
\end{equation}
be the Williamson decomposition of an admissible covariance matrix, and set
\begin{equation}
    x_k=\frac{2\nu_k}{\hbar}\geq1.
\end{equation}
Relative to the trace pairing, the dual Bures quadratic form is represented by
\begin{equation}
    \mathcal G_{{\rm B},\Sigma}(P)
    =8\Sigma P\Sigma+2\hbar^2\Om P\Om.
    \label{eq:multimode-cometric-operator}
\end{equation}

At the Williamson representative \(D\), the cotangent space
\(\operatorname{Sym}(2N)\) decomposes into the following invariant sectors:
\begin{enumerate}
    \item \emph{Local scalar sectors.}  For every mode \(k\), the scalar
    diagonal block has weight
    \begin{equation}
        \lambda_{{\rm rad},k}=2\hbar^2(x_k^2-1),
    \end{equation}
    with total multiplicity \(N\).

    \item \emph{Local traceless sectors.}  Each mode has two traceless
    symmetric directions, both with weight
    \begin{equation}
        \lambda_{{\rm shape},k}=2\hbar^2(x_k^2+1),
    \end{equation}
    with total multiplicity \(2N\).

    \item \emph{Negative-parity cross-mode sectors.}  For each unordered pair
    \(k<l\), there are two independent directions with weight
    \begin{equation}
        \lambda_{-,kl}=2\hbar^2(x_kx_l-1),
    \end{equation}
    with total multiplicity \(N(N-1)\).

    \item \emph{Positive-parity cross-mode sectors.}  For each unordered pair
    \(k<l\), there are two independent directions with weight
    \begin{equation}
        \lambda_{+,kl}=2\hbar^2(x_kx_l+1),
    \end{equation}
    with total multiplicity \(N(N-1)\).
\end{enumerate}
The multiplicities sum to
\begin{equation}
    N+2N+2N(N-1)=N(2N+1)
    =\dim\operatorname{Sym}(2N).
\end{equation}

For a general covariance matrix \(\Sigma=SDS^{\mathsf T}\), these are the
classwise weights in the Williamson-adapted cotangent frame.  The quadratic
forms at \(\Sigma\) and \(D\) are related by an invertible congruence and
therefore have the same inertia and nullity, although their ordinary
eigenvalues relative to a fixed laboratory trace basis need not coincide.

Let
\begin{equation}
    \mathcal P=\{k:x_k=1\},
    \qquad
    m=|\mathcal P|
\end{equation}
be the set and number of pure Williamson modes.  Then
\begin{equation}
    \dim\ker g_{\rm B}^*=m^2.
    \label{eq:multimode-nullity}
\end{equation}
The kernel consists of one local scalar direction for every
\(k\in\mathcal P\), together with two negative-parity cross-mode directions
for every unordered pair \(k<l\) with \(k,l\in\mathcal P\).  All remaining
sectors are strictly positive.  In particular,
\begin{equation}
    \lambda_{{\rm shape},k}\geq4\hbar^2,
    \qquad
    \lambda_{+,kl}\geq4\hbar^2.
\end{equation}
Thus the Bures covariance cometric is nondegenerate in the faithful interior
and loses rank precisely when at least one symplectic eigenvalue reaches the
Williamson floor.
\end{theorem}

\begin{proof}
For a covector \(P\) at \(\Sigma=SDS^{\mathsf T}\), define its
Williamson-frame representative by
\begin{equation}
    \widetilde P=S^{\mathsf T}PS.
\end{equation}
Using \(S^{-1}\Om S^{-\mathsf T}=\Om\) and cyclicity of the trace gives
\begin{equation}
    g_{{\rm B},\Sigma}^*(P,P)
    =g_{{\rm B},D}^*(\widetilde P,\widetilde P).
\end{equation}
The two forms are congruent, so it is sufficient to calculate their inertia
and nullity at \(D\).

Write
\begin{equation}
    J=\begin{pmatrix}0&1\\-1&0\end{pmatrix},
    \qquad
    X=\begin{pmatrix}0&1\\1&0\end{pmatrix},
    \qquad
    Z=\begin{pmatrix}1&0\\0&-1\end{pmatrix}.
\end{equation}
These matrices satisfy
\begin{equation}
    J\Id_2J=-\Id_2,
    \qquad
    JXJ=X,
    \qquad
    JZJ=Z.
\end{equation}
For a diagonal block \(P_{kk}\in\operatorname{Sym}(2)\), the scalar direction
therefore has weight
\begin{equation}
    8\nu_k^2-2\hbar^2=2\hbar^2(x_k^2-1),
\end{equation}
whereas the two traceless directions have weight
\begin{equation}
    8\nu_k^2+2\hbar^2=2\hbar^2(x_k^2+1).
\end{equation}

For \(k<l\), write \(P_{kl}=A\) and \(P_{lk}=A^{\mathsf T}\), with \(A\) an
arbitrary real \(2\times2\) matrix.  The corresponding block is
\begin{equation}
    \bigl(\mathcal G_{{\rm B},D}(P)\bigr)_{kl}
    =8\nu_k\nu_lA+2\hbar^2JAJ.
\end{equation}
The involution \(A\mapsto JAJ\) has negative eigenspace
\(\operatorname{span}\{\Id_2,J\}\) and positive eigenspace
\(\operatorname{span}\{X,Z\}\), each of dimension two.  This gives the two
cross-mode weights stated above.

Because \(x_k\geq1\), a local scalar weight vanishes exactly when \(x_k=1\).
A negative-parity cross-mode weight vanishes exactly when \(x_kx_l=1\), which
on the admissible domain is equivalent to \(x_k=x_l=1\).  No positive-parity
weight can vanish.  If \(m\) modes are pure, the kernel consequently has
dimension
\begin{equation}
    m+2\binom{m}{2}=m^2.
\end{equation}
\end{proof}

\begin{remark}[Fully pure consistency check]
The classwise weights and multiplicities, the nullity \(m^2\) for
\(m=1,2,3\) floor modes, its invariance under a random symplectic frame, and
the identification of the passive and active cross-mode sectors with
beam-splitter and two-mode-squeezing tangents are confirmed numerically in
Ref.~\cite{VerifyScript} (checks~C3, C8, C9).
If all \(N\) modes are pure, Theorem~\ref{thm:multimode-spectrum} gives
nullity \(N^2\) and rank
\begin{equation}
    N(2N+1)-N^2=N(N+1),
\end{equation}
which equals the dimension of the pure Gaussian covariance manifold
\(\operatorname{Sp}(2N,\mathbb R)/U(N)\) \cite{Arvind1995}.
\end{remark}

\begin{remark}[Vacuum-stabilizer structure of the kernel]
\label{rem:kernel-um}
The nullity \(m^2\) has a structural interpretation.  Restrict to the
\(2m\times2m\) block of modes at the Williamson floor, where
\(\Sigma=(\hbar/2)\Id_{2m}\), and denote its symplectic form by \(\Om_m\).
For symmetric covectors supported on this block, the cometric operator becomes
\begin{equation*}
    \mathcal G_{\rm B}(P)
    =2\hbar^2\bigl(P+\Om_mP\Om_m\bigr).
\end{equation*}
Hence \(P\) lies in the kernel exactly when
\(\Om_mP\Om_m=-P\), equivalently \([P,\Om_m]=0\).  The symmetric commutant
has dimension \(m^2\), while the symmetric anticommutant has dimension
\(m(m+1)\).

The map \(P\mapsto\Om_mP\) canonically identifies the commutant with the
vacuum-stabilizer algebra
\begin{equation*}
    \mathfrak u(m)
    =\mathfrak{sp}(2m,\mathbb R)\cap\mathfrak{so}(2m),
\end{equation*}
whereas the anticommutant corresponds to the noncompact tangent sector of
\(\operatorname{Sp}(2m,\mathbb R)/U(m)\).  Within the pure block these two
spaces are trace-orthogonal, so the kernel is the cotangent annihilator of the
pure-orbit tangent space.  Thus
\(\ker g_{\rm B}^*\) is canonically isomorphic to \(\mathfrak u(m)\), while
the complementary \(m(m+1)\)-dimensional sector carries the residual
pure-Gaussian geometry.
\end{remark}

A symplectic-algebraic decomposition of the Gaussian quantum Fisher
information has recently been introduced by Chatterjee \emph{et
al.}~\cite{Chatterjee2026}, who split the QFI additively into an even part
carrying changes of the symplectic spectrum and an odd part associated with
frame-deforming dynamics, using the Cartan decomposition of
\(\mathfrak{sp}(2N,\mathbb R)\) and parity with respect to \(\Om\).  The
classification of Theorem~\ref{thm:multimode-spectrum} is contravariant rather
than covariant and is directed at a different question: it resolves the
cometric into classwise weights with multiplicities in the Williamson-adapted
cotangent frame, and determines the exact nullity of each boundary stratum
when \(m\) of \(N\) modes reach the floor.  The two descriptions are
consistent on the fully pure manifold, where the spectrum-changing sector
degenerates and the residual geometry is that of
\(\operatorname{Sp}(2m,\mathbb R)/U(m)\).

\subsubsection{Auxiliary coupling and the positivity cone}

The lift construction suggests a useful diagnostic deformation: scale the
base--fibre cross coupling by an auxiliary parameter \(\lambda\).  This parameter
is not assigned thermodynamic or dynamical meaning here.  In the convention
with vacuum floor \(\nu=1/2\), define
\begin{equation}
    S_\lambda(P)
    =8\Sigma P\Sigma+2\lambda^2\Om P\Om,
    \qquad S_1=g_{\rm B}^*.
    \label{eq:lambda-family}
\end{equation}
At a one-mode Williamson representative \(\Sigma=\nu\Id_2\), its eigenvalues per
unit Frobenius norm separate into
\begin{equation}
    8\left(\nu^2-\frac{\lambda^2}{4}\right),
    \qquad
    8\left(\nu^2+\frac{\lambda^2}{4}\right),
    \label{eq:lambda-spectrum}
\end{equation}
with multiplicities one and two, respectively.

\begin{proposition}[Auxiliary positivity cone]
\label{prop:positivity-cone}
For the one-mode family \eqref{eq:lambda-family}, positivity holds when
\(\nu>\lambda/2\), the dual form is degenerate when \(\nu=\lambda/2\), and its
formal continuation becomes indefinite when \(\nu<\lambda/2\)
(Ref.~\cite{VerifyScript}, check~C4).  The physical
slice \(\lambda=1\) therefore meets the degeneracy locus exactly at the quantum
uncertainty floor \(\nu=1/2\).
\end{proposition}

\section{Radial dilation and its Bures potential}
\label{sec:potential}

The entropy of one Gaussian normal mode is
\begin{equation}
    s(\nu)
    =\left(\nu+\frac12\right)\log\left(\nu+\frac12\right)
    -\left(\nu-\frac12\right)\log\left(\nu-\frac12\right),
    \label{eq:entropy}
\end{equation}
with
\begin{equation}
    s'(\nu)=\log\frac{\nu+\frac12}{\nu-\frac12}.
\end{equation}
For independent normal modes, \(S=\sum_k s(\nu_k)\).

\begin{proposition}[Exact entropy-scaling identity]
\label{prop:entropy-scaling}
On the fixed-Williamson-frame radial manifold \(\M_{S_W}\),
\begin{equation}
    V(S)=S+\Phi,
    \qquad
    \Phi(\bm\nu)
    =-\frac12\sum_{k=1}^{N}\log\left(\nu_k^2-\frac14\right).
    \label{eq:entropy-scaling}
\end{equation}
\end{proposition}

\begin{proof}
For one mode, direct collection of the coefficients multiplying
\(\log(\nu\pm1/2)\) gives
\begin{equation}
    \nu s'(\nu)-s(\nu)
    =-\frac12\log\left(\nu^2-\frac14\right).
\end{equation}
Summing over modes proves Eq.~\eqref{eq:entropy-scaling}.
\end{proof}

The identity is geometric and algebraic.  A thermodynamic interpretation
requires a specified thermal family.  For a fixed quadratic Hamiltonian with
normal-mode frequencies \(\omega_k\),
\begin{equation}
    \nu_k=\frac12\coth\frac{\beta\omega_k}{2},
    \qquad
    \log\frac{\nu_k+\frac12}{\nu_k-\frac12}=\beta\omega_k.
\end{equation}
Including the zero-point energy,
\begin{equation}
    Z=\prod_k\frac{1}{2\sinh(\beta\omega_k/2)},
\end{equation}
and hence
\begin{equation}
    \Phi
    =\sum_k\log\left[2\sinh\left(\frac{\beta\omega_k}{2}\right)\right]
    =-\log Z=\beta F.
    \label{eq:thermal-phi}
\end{equation}
This interpretation is restricted to fixed-Hamiltonian thermal families; it is
not a claim that \(\Phi\) is a thermodynamic free energy on the entire Gaussian
manifold.

Differentiating Eq.~\eqref{eq:entropy-scaling}, with
\(\mu_k=s'(\nu_k)\), gives
\begin{equation}
    \sum_k\nu_k\,d\mu_k=d\Phi.
    \label{eq:gibbs-duhem-like}
\end{equation}
This is a covariance-scaling analogue of a Gibbs--Duhem relation, distinct from the macroscopic extensive thermodynamic identities that typically generate the Weinhold and Ruppeiner fluctuation geometries \cite{Weinhold1975,Ruppeiner1995}.

\begin{theorem}[Bures-gradient representation]
\label{thm:bures-potential}
On the one-mode radial family at fixed \((r,\varphi)\),
\begin{equation}
    V=-\frac14\grad_{\rm B}\Phi,
    \qquad
    \lVert V\rVert_{\rm B}^2=-\frac14V(\Phi).
    \label{eq:gradient-representation}
\end{equation}
The identities extend additively to \(\M_{S_W}\) whenever the radial QFI is a
direct sum of independent one-mode blocks.
\end{theorem}

\begin{proof}
For one mode,
\begin{equation}
    \partial_\nu\Phi=-\frac{\nu}{\nu^2-\frac14}.
\end{equation}
Using Eq.~\eqref{eq:radial-components},
\begin{equation}
    (\grad_{\rm B}\Phi)^\nu
    =4\left(\nu^2-\frac14\right)
    \left(-\frac{\nu}{\nu^2-\frac14}\right)
    =-4\nu,
\end{equation}
and the shape components vanish.  Thus \(V=-\tfrac14\grad_{\rm B}\Phi\).  Also,
\begin{equation}
    \lVert V\rVert_{\rm B}^2
    =g^{\rm B}_{\nu\nu}\nu^2
    =\frac{\nu^2}{4(\nu^2-\frac14)}
    =-\frac14V(\Phi).
\end{equation}
The fixed-frame multimode statement follows by modewise addition.
\end{proof}

The potential diverges at the boundary,
\begin{equation}
    \Phi(\nu)\longrightarrow+\infty
    \qquad(\nu\downarrow\tfrac12),
\end{equation}
although Proposition~\ref{prop:finite-distance} shows that the boundary is at
finite metric distance.  The potential difference can be interpreted as a
required work only after a force law or control protocol is specified.

It is useful to isolate the floor-dependent scaling response:
\begin{equation}
    V(\Phi)
    =-\sum_k\frac{\nu_k^2}{\nu_k^2-\frac14}
    =-N-\frac14\sum_k\frac{1}{\nu_k^2-\frac14}.
    \label{eq:phi-scaling}
\end{equation}
The second term diverges at the lower floor and vanishes at high symplectic
eigenvalue, anticipating the two regimes studied below.

\section{The intrinsic high-noise classical limit}
\label{sec:upper}

The lower boundary is a singular limit.  The opposite classical regime is
regular.  Restore \(\hbar\) and consider a family whose relevant symplectic
eigenvalues satisfy \(\nu_k\gg\hbar\).  For a fixed covector \(P\), the two terms
in Eq.~\eqref{eq:dual-schematic} scale as
\begin{equation}
    \Tr(P\Sigma P\Sigma)=O(\nu^2\lVert P\rVert_F^2),
    \qquad
    \hbar^2\Tr(P\Om P\Om)=O(\hbar^2\lVert P\rVert_F^2).
\end{equation}

\begin{proposition}[Suppression of the symplectic correction]
\label{prop:upper-limit}
On any spectral sector in which the Fisher-type quadratic form is uniformly
nonzero, the relative symplectic correction to the Gaussian Bures cometric
obeys
\begin{equation}
    \frac{|Q_\Om(P,P)|}
    {g_{\rm FR}^*(\Sigma;P,P)}
    =O\!\left(\frac{\hbar^2}{\nu_{\min}^2}\right)
    \longrightarrow0
    \qquad
    \left(\frac{\nu_{\min}}{\hbar}\to\infty\right).
    \label{eq:upper-suppression}
\end{equation}
Thus the quantum Bures cometric approaches its covariance Fisher component to
leading order.
\end{proposition}

\begin{proof}
The symplectic term is independent of the magnitude of \(\Sigma\) and carries
the factor \(\hbar^2\).  The Fisher-type term is quadratic in \(\Sigma\).  On a
sector bounded away from the kernel of the latter, their ratio is bounded by a
constant times \(\hbar^2/\nu_{\min}^2\).
\end{proof}

The same suppression appears in Eq.~\eqref{eq:phi-scaling}: the floor correction
is \(O(\nu_k^{-2})\).  On a thermal oscillator,
\(\nu\sim(\beta\omega)^{-1}\) at high temperature, so the correction is
\(O((\beta\hbar\omega)^2)\).

There is therefore no finite upper boundary supplied by the Bures metric alone.
The upper regime is asymptotic unless an independent energy, entropy, or
coupling scale is added.  Thermodynamic branch termination in an enlarged lift
is a separate problem and is deliberately excluded here.

\section{Auxiliary comparison with Fisher--Rao and Bures--Wasserstein geometry}
\label{sec:triple-volume}

The preceding sections establish the intrinsic global geometry of the Gaussian
Bures manifold.  This section adds a deliberately separate one-mode comparison
with the classical Fisher--Rao and Bures--Wasserstein cometrics.  None of the
boundary, nullity, radial-potential, or high-noise results depends on the
selection rule introduced below.

Between the singular lower edge and the asymptotic high-noise regime, the three
dual quadratic forms admit a distinguished interior point at which their
determinant densities coincide.  Throughout this section we use the
floor-normalized symplectic eigenvalue
\begin{equation}
    x=\frac{\nu}{\hbar/2}=2\nu\quad(\hbar=1),
    \label{eq:tv-coordinate}
\end{equation}
so that the Williamson floor is \(x=1\) and the relative symplectic correction
of Section~\ref{sec:upper} is \(1/(4\nu^2)=1/x^2\).

\subsection{The three normalized dual determinant densities}

The three determinant densities, the golden-ratio equality, and the quoted
logarithmic slopes below are verified in Ref.~\cite{VerifyScript} (checks~C5,
C6).  All determinants in this section are computed relative to the trace
pairing
\begin{equation}
    \langle P,Q\rangle_{\rm tr}=\Tr(PQ)
\end{equation}
and a basis of \(\operatorname{Sym}(2)\) orthonormal with respect to this
pairing.  This fixes the common reference measure on the three-dimensional
cotangent space.

For the positive covariance cone, let \(\mathcal L_\Sigma(X)\) be the symmetric
solution of
\begin{equation}
    \Sigma\mathcal L_\Sigma(X)+\mathcal L_\Sigma(X)\Sigma=X.
\end{equation}
The Bures--Wasserstein metric and its trace-pairing dual are
\begin{align}
    g_{{\rm BW},\Sigma}(X,Y)
    &=\frac12\Tr\!\left[\mathcal L_\Sigma(X)Y\right],\\
    g_{{\rm BW},\Sigma}^*(P,P)
    &=4\Tr(P\Sigma P),
    \label{eq:bw-dual-exact}
\end{align}
respectively \cite{Malago2018,Bhatia2019,KerskensBWHamiltonian}; the dual form~\eqref{eq:bw-dual-exact} is confirmed independently in Ref.~\cite{VerifyScript} (check~C2).  The same cometric
generates the kinetic term
\begin{equation*}
    \frac12g_{\rm BW}^*(P,P)=2\Tr(P\Sigma P)
\end{equation*}
in the Hamiltonian lift of Bures--Wasserstein covariance dynamics developed in
Ref.~\cite{KerskensBWHamiltonian}.  The factor \(\ell_{\rm BW}\) introduced
below does not alter this tensor structure; it matches the transport action
scale to the dimensionless quantum-information convention used in the present
comparison.  Since this cometric is linear in covariance,
comparison with the quadratic information and quantum Bures cometrics requires
an action scale \(\ell_{\rm BW}>0\).  Define
\begin{equation}
    \widetilde g_{{\rm BW},\ell_{\rm BW}}^*
    =\ell_{\rm BW}g_{\rm BW}^*,
    \qquad
    \eta=\frac{\ell_{\rm BW}}{\hbar}.
    \label{eq:bw-action-matching}
\end{equation}

At an isotropic one-mode state \(\Sigma=\nu\Id_2\),
\begin{equation}
    \widetilde g_{{\rm BW},\ell_{\rm BW}}^*(P,P)
    =4\ell_{\rm BW}\nu\Tr(P^2).
\end{equation}
Its matrix in a trace-orthonormal basis is therefore
\(4\ell_{\rm BW}\nu\Id_3\), and
\begin{align}
    \det\widetilde g_{{\rm BW},\ell_{\rm BW}}^*
    &=(4\ell_{\rm BW}\nu)^3\\
    &=8\hbar^6\eta^3x^3.
    \label{eq:bw-determinant}
\end{align}

The lift-normalized Fisher--Rao information cometric is
\begin{equation}
    \mathcal B_\Sigma(P,P)=8\Tr(P\Sigma P\Sigma).
\end{equation}
At \(\Sigma=\nu\Id_2\), each of its three trace-orthonormal eigenvalues is
\(8\nu^2=2\hbar^2x^2\), and hence
\begin{equation}
    \det\mathcal B_\Sigma=8\hbar^6x^6.
    \label{eq:fr-determinant}
\end{equation}
By Theorem~\ref{thm:null-mobility}, the quantum Bures cometric has one radial
eigenvalue \(2\hbar^2(x^2-1)\) and two shape eigenvalues
\(2\hbar^2(x^2+1)\).  Therefore
\begin{equation}
    \det g_{\rm B}^*
    =8\hbar^6(x^2-1)(x^2+1)^2.
    \label{eq:bures-determinant}
\end{equation}

Removing the common determinant factor \(8\hbar^6\) gives the normalized
dimensionless determinant densities
\begin{equation}
    \Delta_{\rm B}=(x^2-1)(x^2+1)^2,
    \qquad
    \Delta_{\rm FR}=x^6,
    \qquad
    \Delta_{\rm BW}=\eta^3x^3.
    \label{eq:tv-densities}
\end{equation}
The parameter \(\eta\) is a constant matching parameter, not a state variable.
Its normalization status, including the optional boundary-matching convention
\(\eta=1\), is discussed in Section~\ref{sec:matching-scale}.

The equality \(\Delta_{\rm B}=\Delta_{\rm FR}\) reads, with \(u=x^2\),
\begin{equation}
    (u-1)(u+1)^2=u^3
    \quad\Longleftrightarrow\quad
    u^2=u+1,
    \label{eq:tv-golden}
\end{equation}
whose positive root is the golden ratio
\(u=\varphi=(1+\sqrt5)/2\).  Requiring the BW determinant density to have the
same value fixes the second coordinate of the junction:
\begin{equation}
    x=\eta=\sqrt\varphi.
\end{equation}

\begin{remark}[Normalization status of the junction]
\label{rem:tv-invariance}
Within the exact Bures decomposition fixed by
Eq.~\eqref{eq:dual-schematic}, the Bures--Fisher crossing
\(x=\sqrt\varphi\) is invariant under common normalization and reference-basis
changes.  A common rescaling of the Bures cometric and its Fisher block
multiplies both determinants by the same factor, and a common change of basis
multiplies them by the same squared Jacobian.  Their equality is therefore
unchanged.  Equivalently, the reduced radial determinant factor
\((x^2-1)/x^2\) exactly compensates the two enhanced shape factors
\((x^2+1)^2/x^4\) at \(x^2=\varphi\).

This invariance does not make the transport comparison intrinsic.  The BW
cometric is linear rather than quadratic in covariance and requires the
independent action-matching scale \(\eta\).  Nor does the invariant
Bures--Fisher crossing derive the minimum-determinant selection rule.  The full
triple junction is therefore correctly stated as the codimension-two point
\((x,\eta)=(\sqrt\varphi,\sqrt\varphi)\), with \(\eta\) retained as an
explicit constitutive coordinate.
\end{remark}

\subsection{Status of the action-matching scale}
\label{sec:matching-scale}

The parameter \(\eta\) is not fixed by bare covariance geometry for a
principled reason.  The quantum Bures and Fisher--Rao forms measure local
distinguishability and share the normalization inherited from
Eq.~\eqref{eq:dual-schematic}, whereas the Bures--Wasserstein form measures
transport relative to a chosen Euclidean ground cost on phase space.  A
constant rescaling of that ground cost preserves the unparameterized transport
geodesics but rescales its cometric and shifts every determinant crossing with
the information forms.  The symplectic tensor alone supplies no preferred
conversion between these two operational units.

One clean convention is boundary matching.  Requiring the isotropic
Fisher--Rao and action-matched BW cometric eigenvalues to agree at
\(\nu=\hbar/2\) gives
\begin{equation}
    4\ell_{\rm BW}\frac{\hbar}{2}
    =8\left(\frac{\hbar}{2}\right)^2,
\end{equation}
and hence \(\ell_{\rm BW}=\hbar\), or \(\eta=1\).  This makes the scaled
transport cost dimensionless in the adopted units, but it is a normalization
condition rather than a physical derivation.

A physical value of \(\eta\) would require additional constitutive structure,
such as a quadratic Hamiltonian or compatible complex structure fixing the
phase-space ground cost, a specified fluctuation--dissipation relation fixing
transport coefficients, or a common microscopic cotangent action fixing the
relative branch measures.  In a multimode setting the matching may accordingly
be mode dependent or tensorial.  Conversely, the equality at
\(\eta=\sqrt\varphi\) cannot be invoked to determine \(\eta\): choosing the
scale in order to create the triple equality would be circular.  We therefore
retain \(\eta\) as an explicit constitutive parameter.

\subsection{Status of the triple-volume point and open questions}
\label{sec:tv-status}

The triple-volume point
\begin{equation}
    (x,\eta)=(\sqrt{\varphi},\sqrt{\varphi})\approx(1.272,1.272)
    \label{eq:tv-point}
\end{equation}
is the unique kinematic junction of the three one-mode determinant densities,
while its physical and thermodynamic implications remain open.

Geometrically, this junction is the unique one-mode balance point between three
distinct assignments of infinitesimal dual covariance volume.  The Bures
determinant combines one symplectically reduced radial dual weight with two
enhanced shape dual weights; the Fisher--Rao determinant measures the
corresponding covariance-information volume without the symplectic correction;
and the action-matched Bures--Wasserstein determinant measures transport volume
through a different, linear dependence on covariance.  At
\(x=\sqrt\varphi\), the reduced radial factor and the two enhanced shape factors
in the Bures determinant balance so that its normalized determinant density
equals the Fisher--Rao density, while
\(\eta=\sqrt\varphi\) simultaneously brings the transport volume to the same
value.  Thus the point represents equality of three aggregate infinitesimal
volume scales, not equality of the metric tensors or their directional costs.

Under the conditional minimum-log-determinant rule introduced below, this
balance has a simple topological interpretation.  For
\(\eta<\sqrt\varphi\), the transport branch undercuts Fisher--Rao before the
latter can become selected, producing a direct Bures-to-BW crossover.  For
\(\eta>\sqrt\varphi\), a finite Fisher--Rao interval opens between the Bures
and BW regimes.  The triple-volume point is precisely where this intermediate
interval is born with zero width.

At this point, the three
dual determinant densities coincide exactly:
\begin{equation}
    \Delta_{\rm B}=\Delta_{\rm FR}=\Delta_{\rm BW}=\varphi^3,
    \label{eq:triple-volume-equality}
\end{equation}
relative to the common trace-pairing reference measure.

Several geometric properties of this junction are established.  All three dual
forms are positive definite and nonsingular at the intersection.  Only their
determinant densities coincide, however; the metric tensors, eigenvalues, and
directional costs remain distinct.  The crossing is transverse.  At fixed
\(\eta=\sqrt{\varphi}\), the logarithmic slopes are
\begin{equation}
    \left.\partial_x\log\Delta_{\rm B}\right|_{\rm tp}\approx6.06,
    \qquad
    \left.\partial_x\log\Delta_{\rm FR}\right|_{\rm tp}\approx4.72,
    \qquad
    \left.\partial_x\log\Delta_{\rm BW}\right|_{\rm tp}\approx2.36.
    \label{eq:tv-slopes}
\end{equation}
The three branches therefore do not meet tangentially.

Although auxiliary to the intrinsic geometric results established here, the
exact location of the junction and its role in organizing the conditional
branch topology suggest that it may provide a useful landmark for subsequent
extensions.  In particular, it motivates the study of multimode junction loci,
microscopic mechanisms that determine the matching parameter \(\eta\), and
dynamical or thermodynamic models capable of deriving a branch-selection rule.
Whether the junction retains operational significance beyond the present
covariance-level comparison remains open.

\subsection{Conditional branch topology}

\begin{definition}[Conditional minimum-log-determinant rule]
\label{def:minimum-determinant}
At fixed \((x,\eta)\), and relative to the common trace-pairing reference
measure fixed above, the conditionally selected branch is the one with the
smallest normalized determinant density, equivalently the smallest
half-log-determinant.  This is a comparison rule, not a consequence of the
covariance metrics or a derived dynamical law.
\end{definition}

\begin{remark}[Formal cotangent-fluctuation interpretation]
\label{rem:fluctuation-interpretation}
The comparison rule has an exact Gaussian-integral interpretation.  For each
positive branch cometric, define on the common trace-normalized cotangent space
\begin{equation*}
    Z_k^{(P)}(T_k,w_k)
    =w_k\int
    \exp\!\left[-\frac{g_k^*(P,P)}{2T_k}\right]dP
    =w_k(2\pi T_k)^{3/2}\bigl(\det g_k^*\bigr)^{-1/2}.
\end{equation*}
If the channel scales \(T_k\) and weights \(w_k\) are equal, maximizing this
formal cotangent weight is equivalent to minimizing \(\log\Delta_k\).  Unequal
scales or weights replace the comparison by minimization of
\begin{equation*}
    \log\Delta_k-3\log T_k-2\log w_k.
\end{equation*}
The equivalence of the three formulations---smallest determinant density, largest cotangent weight, and largest Gaussian entropy---together with the displacement of the junction under unequal \(T_k\), is verified in Ref.~\cite{VerifyScript} (check~C10).

This identity interprets the rule but does not derive a physical competition
between the branches.  In particular, operational divergences on state space have the primal metrics as their Hessians, up to the standard convention factors (verified in Ref.~\cite{VerifyScript}, check~C11, for quantum infidelity, relative entropy, and the quadratic Wasserstein cost); integrating tangent fluctuations
would instead produce \((\det g_k)^{-1/2}=(\det g_k^*)^{1/2}\), with the
opposite cometric ordering.  A microscopic realization of the present rule
must therefore supply common cotangent variables with quadratic actions
\(g_k^*(P,P)/2\), together with their relative scales and weights.  The
Hamiltonian BW lift of Ref.~\cite{KerskensBWHamiltonian} provides such a kinetic
interpretation for the BW branch alone, not for a coupled three-branch system.
\end{remark}

\begin{remark}[Status of the comparison premise]
\label{rem:comparison-premise}
Two parts of the determinant comparison are intrinsic.  The three forms act on
the same covariance cotangent space, and determinant ratios computed with a
common trace-pairing measure are basis independent.  The Bures and Fisher--Rao
normalizations are also fixed relative to one another by the exact decomposition
in Eq.~\eqref{eq:dual-schematic}; the additional transport normalization is the
explicit parameter \(\eta\).

What is not intrinsic is the assertion that the three forms occur as competing
channels in one physical ensemble with equal weights, equal fluctuation scales,
and a common large parameter.  Copy number or sample number can control
information fluctuations, while inverse noise strength or mode number can
control transport fluctuations, but no model making one parameter govern all
three is constructed here.  Likewise, covariance kinematics contains no
transition mechanism between the branches.  Along an externally driven path,
a determinant crossing is therefore a change in the conditionally preferred
cost functional, not by itself a dynamical event.
\end{remark}

\begin{proposition}[Conditional branch topology]
\label{prop:branch-topology}
Assume the conditional minimum-log-determinant rule of
Definition~\ref{def:minimum-determinant}.  Let \(x>1\) and \(\eta>0\).  Then
the conditionally selected one-mode branch diagram has two possible
topologies, separated by
\begin{equation}
    \eta_c=\sqrt{\varphi}.
\end{equation}

\begin{enumerate}
    \item \emph{Direct crossover.} If \(0<\eta<\sqrt{\varphi}\), the
    Fisher--Rao branch is never the strict global minimizer.  There is a unique
    point \(x_c(\eta)>1\) satisfying
    \begin{equation}
        \Delta_{\rm B}\bigl(x_c(\eta)\bigr)
        =\Delta_{\rm BW}\bigl(x_c(\eta);\eta\bigr),
    \end{equation}
    with
    \begin{equation}
        1<x_c(\eta)<\sqrt{\varphi}.
    \end{equation}
    The selected sequence is
    \begin{equation}
        {\rm Bures}\longrightarrow{\rm BW}.
    \end{equation}

    \item \emph{Sequential crossover.} If \(\eta>\sqrt{\varphi}\), the
    Fisher--Rao branch is the strict global minimizer on the nonempty interval
    \begin{equation}
        \sqrt{\varphi}<x<\eta.
    \end{equation}
    The selected sequence is
    \begin{equation}
        {\rm Bures}\longrightarrow{\rm Fisher\text{--}Rao}
        \longrightarrow{\rm BW},
    \end{equation}
    with crossover points
    \begin{equation}
        x_{c1}=\sqrt{\varphi},
        \qquad
        x_{c2}=\eta.
    \end{equation}
\end{enumerate}

At \(\eta=\sqrt{\varphi}\), the intermediate Fisher--Rao interval collapses
to zero width.  The three determinant densities then coincide at the isolated
junction in Eq.~\eqref{eq:tv-point}.  Thus the triple-volume point is the
codimension-two boundary between the direct and sequential conditional branch
diagrams (Fig.~\ref{fig:conditional-branches}).
\end{proposition}

\noindent The equalities, exact crossover locations, and selected sequences on
both sides of \(\eta=\sqrt{\varphi}\) are confirmed on a dense numerical grid
in Ref.~\cite{VerifyScript} (check~C7).

\begin{proof}
The Bures--Fisher equality is
\begin{equation}
    (x^2-1)(x^2+1)^2=x^6,
\end{equation}
or equivalently
\begin{equation}
    x^4-x^2-1=0.
\end{equation}
Its unique solution in the physical interior is \(x=\sqrt{\varphi}\).  Hence
\begin{equation}
    \Delta_{\rm B}<\Delta_{\rm FR}
    \quad\text{for}\quad 1<x<\sqrt{\varphi},
\end{equation}
whereas
\begin{equation}
    \Delta_{\rm FR}<\Delta_{\rm B}
    \quad\text{for}\quad x>\sqrt{\varphi}.
\end{equation}

The Fisher--BW equality is
\begin{equation}
    x^6=\eta^3x^3
    \quad\Longleftrightarrow\quad x=\eta,
\end{equation}
for \(x>0\).  Thus
\begin{equation}
    \Delta_{\rm FR}<\Delta_{\rm BW}
    \quad\Longleftrightarrow\quad x<\eta.
\end{equation}
The Fisher--Rao branch is therefore the strict global minimizer precisely when
\begin{equation}
    \sqrt{\varphi}<x<\eta,
\end{equation}
which is a nonempty interval if and only if \(\eta>\sqrt{\varphi}\).  This
identifies the intermediate branch in the sequential case.  The ordering of
the two exterior intervals follows from the Bures--BW comparison below.

For the direct Bures--BW comparison, define
\begin{equation}
    h(x)=\frac{\Delta_{\rm B}(x)}{x^3}
    =\frac{(x^2-1)(x^2+1)^2}{x^3}.
\end{equation}
The crossing equation is \(h(x)=\eta^3\).  For \(x>1\),
\begin{equation}
    \frac{d}{dx}\log h(x)
    =\frac{3x^4-2x^2+3}{x(x^4-1)}>0.
\end{equation}
Moreover,
\begin{equation}
    \lim_{x\downarrow1}h(x)=0,
    \qquad
    \lim_{x\to\infty}h(x)=\infty.
\end{equation}
The crossing consequently exists and is unique.  At
\(x=\sqrt{\varphi}\),
\begin{equation}
    h(\sqrt{\varphi})=\varphi^{3/2}
    =(\sqrt{\varphi})^3.
\end{equation}
Strict monotonicity therefore gives
\begin{equation}
    x_c(\eta)<\sqrt{\varphi}
    \quad\Longleftrightarrow\quad
    \eta<\sqrt{\varphi}.
\end{equation}
In the direct regime one also has \(x_c(\eta)>\eta\).  If
\(0<\eta\leq1\), this follows immediately from \(x_c(\eta)>1\).  If
\(1<\eta<\sqrt{\varphi}\), both arguments lie in the domain of strict
monotonicity of \(h\), and
\begin{equation}
    h(\eta)-\eta^3
    =\eta-\eta^{-1}-\eta^{-3}<0
\end{equation}
because \(\eta^4-\eta^2-1<0\).  Hence \(x_c(\eta)>\eta\) throughout the
direct regime and \(\Delta_{\rm FR}>\Delta_{\rm BW}\) for
\(x>x_c(\eta)\).  Bures is selected below \(x_c(\eta)\), BW is selected above
it, and Fisher--Rao is never the strict global minimizer.  At
\(\eta>\sqrt{\varphi}\), monotonicity gives
\(\Delta_{\rm B}<\Delta_{\rm BW}\) for \(1<x<\sqrt{\varphi}\), so Bures is
selected there; for \(x>\eta\), one has
\(\Delta_{\rm BW}<\Delta_{\rm FR}<\Delta_{\rm B}\), so BW is selected.
This completes the sequential ordering.  At \(\eta=\sqrt{\varphi}\), both
crossover locations equal
\(x=\sqrt{\varphi}\), where Eq.~\eqref{eq:triple-volume-equality} holds.
\end{proof}

\begin{remark}[Relative tensor convergence versus conditional branch selection]
\label{rem:tensor-versus-selection}
It is important to distinguish the asymptotic equivalence of two cometric
tensors from the conditional selection of a determinant branch.  Write
\begin{equation}
    g_{\rm B}^*=\mathcal B_\Sigma+Q_\Om,
\end{equation}
where
\begin{equation}
    \mathcal B_\Sigma(P,P)=8\Tr(P\Sigma P\Sigma)
\end{equation}
is the lift-normalized covariance Fisher--Rao cometric and
\begin{equation}
    Q_\Om(P,P)=2\hbar^2\Tr(P\Om P\Om)
\end{equation}
is the state-independent symplectic correction.

In the high-noise regime, on every spectral sector satisfying the uniformity
condition of Proposition~\ref{prop:upper-limit},
\begin{equation}
    \frac{|Q_\Om(P,P)|}{\mathcal B_\Sigma(P,P)}
    =O\!\left(\frac{\hbar^2}{\nu_{\min}^2}\right)
    \longrightarrow0.
\end{equation}
Equivalently, in the one-mode isotropic sector,
\begin{equation}
    \mathcal B_\Sigma^{-1/2}\mathcal G_{{\rm B},\Sigma}
    \mathcal B_\Sigma^{-1/2}
    \longrightarrow\operatorname{Id}_{\operatorname{Sym}(2)}
    \qquad(x\to\infty).
\end{equation}
Thus the quantum Bures cometric becomes asymptotically equivalent, in this
relative sense, to its covariance Fisher--Rao component.  The unscaled
cometrics do not converge in absolute operator norm: their leading terms grow
with covariance while their difference remains of order \(\hbar^2\).

Neither \(g_{\rm B}^*\) nor \(\mathcal B_\Sigma\) converges tensorially to the
Bures--Wasserstein transport cometric.  The information cometrics are quadratic
in \(\Sigma\), whereas the Bures--Wasserstein cometric is linear in \(\Sigma\)
and requires the independent action-matching parameter \(\eta\).  They
therefore retain different tensor structures and operational meanings.

For every fixed \(\eta>0\), the Bures--Wasserstein determinant density
eventually becomes the smallest of the three because
\begin{equation}
    \Delta_{\rm BW}=O(x^3),
    \qquad
    \Delta_{\rm FR}=O(x^6),
    \qquad
    \Delta_{\rm B}=O(x^6).
\end{equation}
Its selection at sufficiently large \(x\) is therefore not a metric limit.  It
is a conditional switch of the selected determinant branch under
Definition~\ref{def:minimum-determinant}.  Without a microscopic derivation of
that selection rule, this branch switch should not be identified with a
dynamical or thermodynamic phase transition.
\end{remark}

\begin{figure}[t]
\centering
\resizebox{0.82\columnwidth}{!}{%
\begin{tikzpicture}[x=4.5cm,y=1.2cm,>=Stealth,font=\footnotesize]
  \def\xt{1.27202}
  \def\top{2.35}

  \fill[blue!13]
    (1,0) -- (1.03,0.6186) -- (1.06,0.7764) -- (1.09,0.8860)
    -- (1.12,0.9727) -- (1.15,1.0458) -- (1.18,1.1098)
    -- (1.21,1.1673) -- (1.24,1.2200) -- (\xt,\xt)
    -- (\xt,\top) -- (1,\top) -- cycle;
  \fill[orange!16]
    (1,0) -- (1.03,0.6186) -- (1.06,0.7764) -- (1.09,0.8860)
    -- (1.12,0.9727) -- (1.15,1.0458) -- (1.18,1.1098)
    -- (1.21,1.1673) -- (1.24,1.2200) -- (\xt,\xt)
    -- (2.35,2.35) -- (2.35,0) -- cycle;
  \fill[green!13]
    (\xt,\xt) -- (2.35,2.35) -- (\xt,\top) -- cycle;

  \draw[very thick,blue!65!black]
    (1,0) -- (1.03,0.6186) -- (1.06,0.7764) -- (1.09,0.8860)
    -- (1.12,0.9727) -- (1.15,1.0458) -- (1.18,1.1098)
    -- (1.21,1.1673) -- (1.24,1.2200) -- (\xt,\xt);
  \draw[very thick,green!45!black] (\xt,\xt) -- (\xt,\top);
  \draw[very thick,orange!70!black] (\xt,\xt) -- (2.35,2.35);

  \draw[->] (0.96,0) -- (2.43,0) node[right] {$x=2\nu/\hbar$};
  \draw[->] (1,0) -- (1,2.47) node[above] {$\eta$};
  \draw (\xt,0.035) -- (\xt,-0.035)
    node[below=2pt] {$\sqrt\varphi$};
  \draw (1.012,\xt) -- (0.988,\xt)
    node[left=2pt] {$\sqrt\varphi$};

  \fill (\xt,\xt) circle (1.5pt);
  \node[anchor=south west,align=left] at (\xt+0.0,\xt-0.35)
    {triple-point};
  \node[blue!55!black] at (1.12,1.65) {Bures};
  \node[green!40!black,rotate=0] at (1.52,1.92) {Fisher--Rao};
  \node[orange!70!black] at (1.75,0.8) {BW};
\end{tikzpicture}%
}
\caption{Conditional one-mode branch topology in the \((x,\eta)\) plane.
The colored regions show the smallest normalized dual determinant density only
under Definition~\ref{def:minimum-determinant}.  The Bures--BW boundary is the
unique curve \(h(x)=\eta^3\) for \(1<x<\sqrt\varphi\); the Bures--Fisher and
Fisher--BW boundaries are \(x=\sqrt\varphi\) and \(\eta=x\), respectively.
Their codimension-two junction is
\((x,\eta)=(\sqrt\varphi,\sqrt\varphi)\).  The diagram is a conditional
determinant comparison, not a thermodynamic phase diagram.}
\label{fig:conditional-branches}
\end{figure}
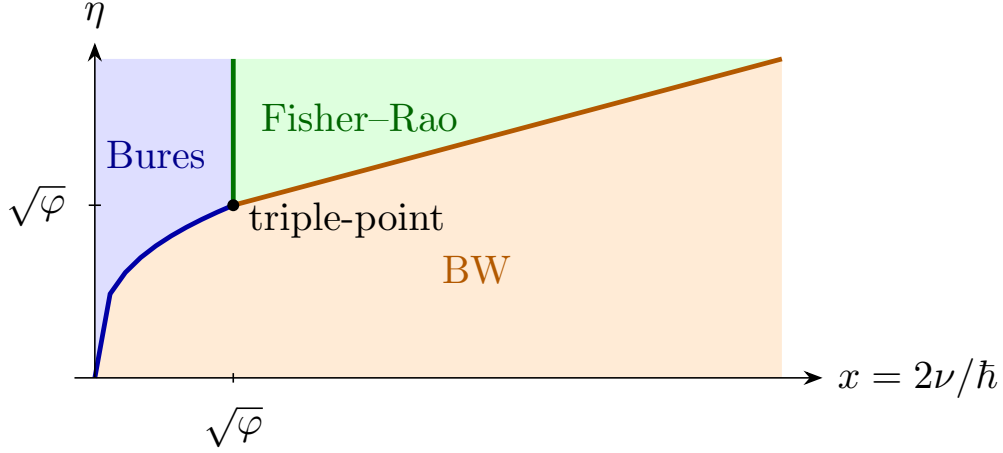

The triple-volume point is not part of the asymptotic high-noise limit.  At
\(x=\sqrt{\varphi}\), the magnitude of the symplectic correction relative to
the information-block weight is
\begin{equation}
    \frac{1}{x^2}=\frac{1}{\varphi}\approx0.618.
    \label{eq:tv-correction}
\end{equation}
The correction is therefore still of order unity.  Moreover, the symplectic
eigenvalue is only a factor \(\sqrt{\varphi}\approx1.272\) above the Williamson
floor.  The triple-volume equality should consequently not be interpreted as an
asymptotic classical limit.

Despite these exact kinematic results, several questions remain open:

\begin{enumerate}
\item \emph{Physical realizability.}
It is not known whether a microscopic system selects
\(\eta=\sqrt{\varphi}\).  The boundary-matching convention instead gives
\(\eta=1\).  A Hamiltonian ground cost, thermalizing bath, or microscopic
fluctuation theory is required to determine whether the triple value can be
physically realized.

\item \emph{Branch-selection dynamics.}
Remark~\ref{rem:fluctuation-interpretation} gives an exact formal cotangent
interpretation of the minimum-log-determinant rule, but no common physical
channel ensemble is derived.  A coupled cotangent lift or specified
detailed-balance model would have to fix the three quadratic actions, their
relative weights and scales, and the transition mechanism between them.

\item \emph{Thermodynamic status.}
No associated thermodynamic free energy, latent heat, critical exponent,
susceptibility divergence, or coexistence law has been demonstrated.
Equation~\eqref{eq:triple-volume-equality} therefore does not establish a
thermodynamic triple point, critical endpoint, or tricritical point.

\item \emph{Operational signature.}
Equal determinant densities do not imply equal quantum Fisher information
in a specified direction, equal transport cost, equal dissipation rate, or
equal geodesic length.  An experimentally accessible observable that detects
passage through the junction remains to be identified.

\item \emph{Multimode survival.}
The golden-ratio value follows from the one-mode determinant equation, with
one radial and two shape directions, together with the adopted branch
normalization.  For unequal multimode symplectic spectra, the junction may
split into a higher-dimensional locus, become mode dependent, or disappear.

\item \emph{Robustness beyond quadratic order.}
Non-Gaussian fluctuation corrections, branch-dependent measures, or
higher-order terms in an effective action may shift the equality point.
The location \(\sqrt{\varphi}\) is exact within the present one-mode
quadratic comparison, but its universality beyond that model is unknown.

\item \emph{Geometric interpolation.}
The present framework compares three separate metrics.  It does not
construct a single effective metric that interpolates continuously among
the Bures, Fisher--Rao, and Bures--Wasserstein structures.  Whether such an
extended geometry exists remains an open mathematical problem.

\end{enumerate}

\subsection{Kinematic crossover and finite-rate radial cost}
\label{sec:kinematic-dynamical}

The determinant comparison above is kinematic and conditional.  Under the
minimum-log-determinant rule, or within the formal cotangent Gaussian ensemble
of Remark~\ref{rem:fluctuation-interpretation}, it compares aggregate
dual-volume weights assigned to three covariance geometries.  It does not
provide an equation of motion or a mechanism that transfers a system between
the corresponding branches.

A separate comparison can be made in the primal metrics.  Given the same
externally prescribed radial trajectory, the Bures and Bures--Wasserstein
metrics assign different kinetic costs to its execution.  The following result
uses the convention \(\hbar=1\).

\begin{proposition}[Finite-rate radial cost separation]
\label{prop:finite-rate-cost}
Let \(\Sigma(t)=\nu(t)\Id_2\) be a continuously differentiable radial
compression trajectory with
\begin{equation}
    \nu(t)\downarrow\frac12
\end{equation}
and suppose that, sufficiently close to the boundary,
\begin{equation}
    0<v_0\leq-\dot\nu(t)\leq v_1<\infty.
    \label{eq:finite-rate-assumption}
\end{equation}
Define the Bures and Bures--Wasserstein kinetic densities by
\begin{equation}
    K_{\rm B}=g^{\rm B}_{\nu\nu}\dot\nu^{,2},
    \qquad
    K_{\rm BW}=g_{\rm BW}(\dot\Sigma,\dot\Sigma).
\end{equation}
Then
\begin{equation}
    K_{\rm B}
    =\frac{\dot\nu^{,2}}
    {4\left(\nu^2-\frac14\right)}
    \longrightarrow\infty
    \qquad\left(\nu\downarrow\frac12\right),
    \label{eq:finite-rate-bures-density}
\end{equation}
whereas
\begin{equation}
    K_{\rm BW}=\frac{\dot\nu^{,2}}{2\nu}
    \label{eq:finite-rate-bw-density}
\end{equation}
remains bounded.  Moreover, over the final segment approaching the boundary,
the integrated Bures kinetic action diverges while the integrated
Bures--Wasserstein kinetic action remains finite (the divergent Bures and
finite Bures--Wasserstein radial costs, and the deceleration bound of
Corollary~\ref{cor:finite-budget-deceleration}, are checked numerically in
Ref.~\cite{VerifyScript}, check~C13):
\begin{equation}
    \int K_{\rm B}\,dt=\infty,
    \qquad
    \int K_{\rm BW}\,dt<\infty.
    \label{eq:finite-rate-actions}
\end{equation}
\end{proposition}

\begin{proof}
The Bures expression follows from Eq.~\eqref{eq:radial-components}.  For the
Bures--Wasserstein metric, the Lyapunov equation
\begin{equation}
    \Sigma\mathcal L_\Sigma(\dot\Sigma)
    +\mathcal L_\Sigma(\dot\Sigma)\Sigma
    =\dot\Sigma
\end{equation}
gives
\begin{equation}
    \mathcal L_\Sigma(\dot\Sigma)
    =\frac{\dot\nu}{2\nu}\Id_2.
\end{equation}
Using the convention of Eq.~\eqref{eq:bw-dual-exact},
\begin{equation}
    g_{\rm BW}(\dot\Sigma,\dot\Sigma)
    =\frac12\Tr\!\left[
      \mathcal L_\Sigma(\dot\Sigma)\dot\Sigma
      \right]
    =\frac{\dot\nu^{,2}}{2\nu}.
\end{equation}
The bounds in Eq.~\eqref{eq:finite-rate-assumption} make the Bures density
diverge and keep the BW density bounded.  Since the path is strictly monotone,
changing variables from \(t\) to \(\nu\) gives
\begin{align}
    \int K_{\rm B}\,dt
    &\geq\frac{v_0}{4}
      \int_{1/2}^{\nu_1}
      \frac{d\nu}{\nu^2-\frac14}=\infty,\\
    \int K_{\rm BW}\,dt
    &\leq\frac{v_1}{2}
      \int_{1/2}^{\nu_1}\frac{d\nu}{\nu}<\infty,
\end{align}
where \(\nu_1>1/2\) lies in the final segment on which the rate bounds hold.
\end{proof}

\begin{corollary}[Finite-budget radial deceleration]
\label{cor:finite-budget-deceleration}
If the instantaneous Bures kinetic density is bounded by
\(K_{\rm B}\leq K_{\max}\), then
\begin{equation}
    |\dot\nu|
    \leq2\sqrt{K_{\max}\left(\nu^2-\frac14\right)}
    \longrightarrow0
    \qquad\left(\nu\downarrow\frac12\right).
    \label{eq:finite-budget-rate}
\end{equation}
Thus a bounded-budget Bures trajectory must decelerate in the covariance
coordinate on approach to the Williamson floor.
\end{corollary}

\begin{remark}[No forced transport conclusion]
Proposition~\ref{prop:finite-rate-cost} compares the costs assigned by two
metrics to the same externally prescribed finite-rate trajectory.  It does not
derive a transition from Bures to Bures--Wasserstein geometry.  A decelerating
schedule can have finite Bures kinetic density and finite integrated action;
for example, \(\nu(t)=\tfrac12+(T-t)^2\) reaches the boundary with bounded
\(K_{\rm B}\).  Nor does the Bures--Wasserstein metric provide a quantum
continuation into \(\nu<1/2\), since such covariances lie outside the quantum
Gaussian state space.  Selecting a transport branch requires an independent
dynamical rule beyond the cost separation proved here.
\end{remark}

\section{Interpretation, limitations, and outlook}
\label{sec:two-routes}

\subsection{Boundary restriction versus asymptotic classicality}

The global picture is summarized by two regimes of unequal logical status
(Fig.~\ref{fig:signature}):
\begin{enumerate}
    \item \emph{Quantum boundary restriction.}  As
    \(\nu_{\min}\downarrow\hbar/2\), the faithful radial metric becomes
    singular and the dual radial mobility vanishes.  Tangential pure-Gaussian
    geometry remains finite.  This is not a classical limit.  A classical
    statistical description may arise only after restriction to an externally
    specified commutative measurement or conditioned submanifold.
    \item \emph{Intrinsic asymptotic suppression.}  As
    \(\nu_{\min}/\hbar\to\infty\), the symplectic correction becomes negligible
    relative to the Fisher block.  The geometry remains nondegenerate and
    approaches classical covariance statistics smoothly.
\end{enumerate}

Only the upper regime is therefore an intrinsic classical limit of the metric.
At the Williamson floor, the intrinsic boundary geometry is the
Fubini--Study geometry of the pure Gaussian orbit, and the cometric kernel is
canonically isomorphic to the vacuum-stabilizer algebra \(\mathfrak u(m)\).
The lower regime is more accurately described as a quantum boundary from which
classical statistics may be extracted conditionally; the covariance geometry
supplies neither the measurement map nor the conditioning rule.

\begin{table}[b]
\caption{The two regimes and their unequal logical status.  Only the
high-noise regime is an intrinsic classical limit of the Bures geometry.}
\label{tab:two-limits}
\begin{ruledtabular}
\begin{tabular}{lll}
Regime & Geometric mechanism & Established geometry\\
\hline
\(\nu_{\min}\downarrow\hbar/2\) & Radial degeneration and boundary restriction
& Pure Gaussian / conditionally restricted statistics\\
\(\nu_{\min}/\hbar\to\infty\) & Relative symplectic suppression
& Covariance Fisher component\\
\end{tabular}
\end{ruledtabular}
\end{table}

\begin{figure}[t]
\centering
\resizebox{\columnwidth}{!}{%
\begin{tikzpicture}[>=Stealth, line join=round, font=\footnotesize]
  \def\Rf{2}       
  \def\Rout{5.7}
  \shade[inner color=blue!16, outer color=white] (0,0) circle (\Rout);
  \fill[red!6] (0,0) circle (\Rf);
  \fill[pattern=north east lines, pattern color=red!28] (0,0) circle (\Rf);
  \draw[densely dashed, very thick, red!70!black] (0,0) circle (\Rf);
  \draw[->, thick] (0,0) -- (\Rout+0.6,0) node[right] {$\nu$};
  \fill (0,0) circle (1.5pt);   \node[above left=-2pt] at (0.4,-0.4) {$\nu=0$};
  \draw[thick] (\Rf,0.13) -- (\Rf,-0.13);  \node[below=2pt] at (\Rf,-0.13) {$\nu=\tfrac12$};
  \draw[thick] (4,0.10) -- (4,-0.10);       \node[below=2pt] at (4,-0.10) {$\nu=1$};
  \node[below right, align=left, gray!55!black] at (4.4,-0.5)
     {radial / dilation axis:\\ metric singular at floor,\\ dual mobility $\to 0$};
  \draw[->, thick] (26:\Rf+2.15) -- (26:\Rf+0.10);
  \node[anchor=west, align=left] at (26:\Rf+2.2)
     {backward dilation $D_s$\\ meets floor at finite $s_*$};
  \draw[<->, very thick, blue!60!black]
       (110:\Rf) arc[start angle=110, end angle=160, radius=\Rf];
  \node[blue!45!black, align=center, anchor=east] at (152:\Rf+1.9)
     {pure-Gaussian orbit\\ (squeeze, rotation):\\ tangent geometry regular};
  \node[red!55!black, align=center] at (0,.8) {quantum forbidden\\ $\nu<\tfrac12$};
  \node[align=center] at (-122:\Rf+1.55) {admissible mixed\\ Gaussian states\\ $\nu>\tfrac12$};
  \node[align=center, blue!40!black] at (54:\Rout-1.055)
     {high noise:\\ $O(\hbar^2/\nu^2)\!\to\!0$\\ (classical Fisher)};
  \node[align=center, red!45!black, anchor=north] at (0,-\Rout-0.05)
     {classical covariance cone continues to $\nu=0$ (no quantum floor)};
  \node[align=center, blue!45!black, anchor=south] at (0,\Rout+0.05)
     {no finite upper boundary --- admissible region open as $\nu\to\infty$};
\end{tikzpicture}%
}
\caption{Phase-domain schematic of the one-mode Gaussian Bures manifold
(radius \(\propto\nu\); \(\hbar=1\), floor \(\nu=\tfrac12\)).  The classical
covariance cone extends to \(\nu=0\); the quantum-admissible subdomain is
\(\nu>\tfrac12\).  The dashed Williamson edge is the pure-state boundary: the
radial (dilation) direction is singular there in the primal metric and null in
the dual [Eq.~\eqref{eq:radial-components}, Thm.~\ref{thm:null-mobility}], while
pure-Gaussian squeeze/rotation directions tangent to the edge remain regular
[Thm.~\ref{thm:anisotropic-boundary}].  Backward dilation \(D_s\) meets the edge
at finite parameter \(s_*\) [Prop.~\ref{prop:finite-contact}].  At large \(\nu\)
the symplectic correction is suppressed as \(O(\hbar^2/\nu^2)\) and the geometry
approaches its covariance Fisher component [Prop.~\ref{prop:upper-limit}]; there
is no finite upper boundary.}
\label{fig:signature}
\end{figure}
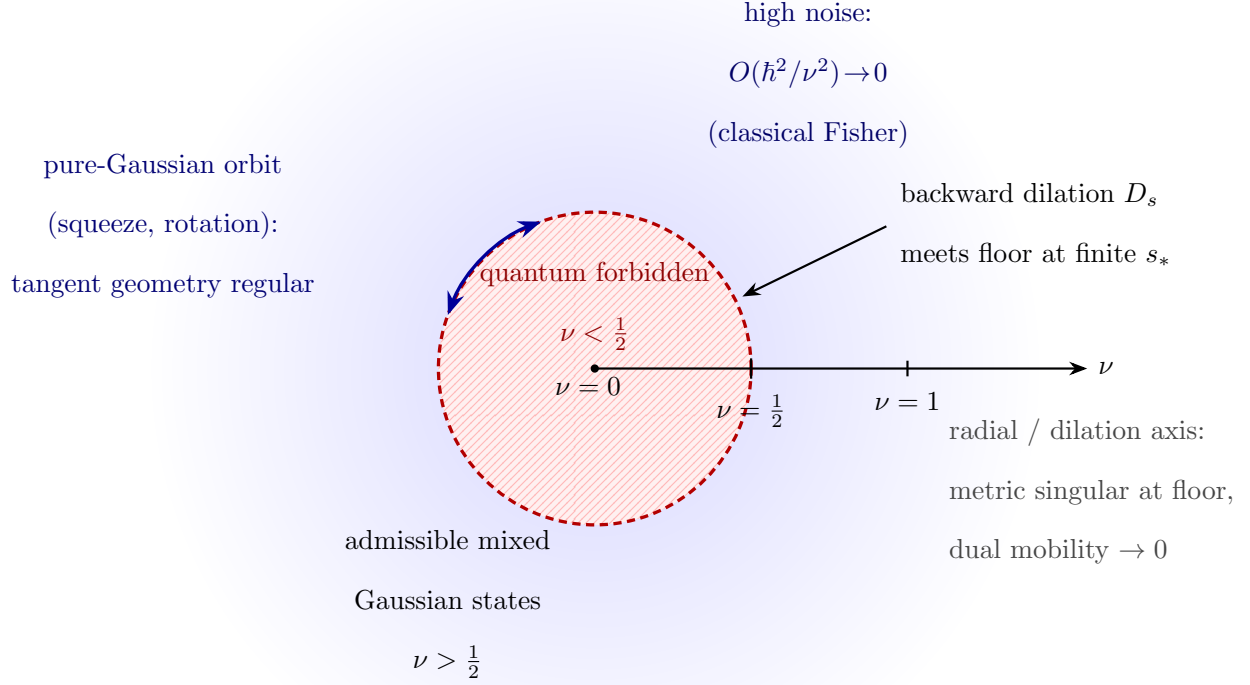

The terminology ``quantum band'' can therefore be used only qualitatively.  The
physical Bures manifold occupies \(\nu_k>\hbar/2\), has a singular lower edge,
and an asymptotically classical large-\(\nu\) region.  The triple-volume
junction of Section~\ref{sec:triple-volume} sits strictly between the two
regimes, at \(x=\sqrt{\varphi}\), where the symplectic correction is still of
order unity; it is a kinematic crossing, not a boundary.  Without additional
physics there is no second finite spectral wall.

\subsection{Measurement, commutative restriction, and boundary statistics}
\label{sec:measurement}

The boundary and high-noise regimes provide two inequivalent routes to
classical statistical geometry.  In the high-noise regime, suppression of the
symplectic correction is intrinsic: the Gaussian Bures geometry approaches its
covariance Fisher component without selecting a measurement.  At the
Williamson floor, by contrast, the surviving state-space geometry remains
quantum and is tangentially the geometry of the pure Gaussian orbit.  A
classical statistical model arises there only after a measurement channel,
commuting observable algebra, or conditioning rule maps the quantum family to
a family of outcome distributions.

This gives measurement a precise but limited role.  The covariance geometry
determines the admissible domain, the anisotropic boundary structure, and the
directions that remain regular or become null.  A specified measurement then
determines which of those directions are statistically accessible in its
classical outcome space.  The metric alone does not choose that measurement,
derive its stochastic record, or generate collapse; nevertheless, it supplies
the geometric structure on which any such operational restriction must act.

\subsection{Gaussian versus non-Gaussian continuation}
\label{sec:continuation}

The spectral floor and departure from Gaussianity are distinct events.  The
Gaussian manifold is characterized not only by its covariance but also by the
vanishing of all cumulants above second order.  A Gaussian and a non-Gaussian
state can share the same first and second moments; covariance geometry alone
cannot distinguish them.

The covariance nullity established by Theorem~\ref{thm:null-mobility} does not
imply that every direction in the full quantum-state space becomes singular.

\begin{remark}[Transverse non-Gaussian tangents]
\label{rem:non-gaussian-tangents}
Let \(|\psi_G\rangle\) be a pure Gaussian state and let
\(|\chi_{\rm NG}\rangle\) be a fixed normalized vector orthogonal both to
\(|\psi_G\rangle\) and to the tangent space of its pure-Gaussian unitary orbit.
The projective curve
\begin{equation}
    |\psi_\epsilon\rangle
    =\frac{|\psi_G\rangle+\epsilon|\chi_{\rm NG}\rangle}
    {\sqrt{1+\epsilon^2}}
\end{equation}
has a finite, nonzero Fubini--Study norm at \(\epsilon=0\).  Regular
non-Gaussian directions therefore exist in the full pure-state manifold even
though the radial Gaussian covariance direction becomes singular.

Whether an actual evolution develops a component along such a transverse
direction depends on its generator.  A specified Lindblad equation,
measurement-conditioning rule, or nonlinear Hamiltonian can be tested for this
behavior, but the covariance metric alone neither produces nor selects
non-Gaussian escape.
\end{remark}

Gaussian-preserving dynamics---for example quadratic Hamiltonians with linear
Gaussian noise or Gaussian conditioning---can approach the floor and remain
Gaussian.  A genuine non-Gaussian transition requires a nonquadratic generator,
non-Gaussian measurement, or loss of transverse stability in higher-cumulant
variables.

\subsection{Scope, limitations, and open questions}
\label{sec:discussion}

The construction of the Gaussian Bures metric and the study of its global
geometry are distinct problems.  The companion lift identifies the algebraic
origin of the symplectic correction, while Section~\ref{sec:manifold} derives
the dual form needed here directly from the Gaussian SLD equation.  That form
exposes the domain on which the correction yields a positive Riemannian
geometry, its loss of rank at the Williamson floor, and its relative
disappearance at high noise.

The most important boundary feature is anisotropy.  The divergence of
\(g^{\rm B}_{\nu\nu}\) is not a divergence of every direction and does not erase
the pure Gaussian manifold.  It is a singularity of the faithful-state radial
chart associated with changing rank.  In the dual picture, the corresponding
radial coefficient vanishes; under a specified gradient-flow interpretation,
this coefficient acts as a vanishing radial mobility.  Geometrically, the
floor admits no positive-definite continuation in the radial compression
direction, while unitary shape motion remains regular.

The primal comparison adds a deliberately conditional dynamical statement.
If an external protocol insists on a radial coordinate rate bounded away from
zero, the Bures kinetic density and integrated action diverge, while the
Bures--Wasserstein cost assigned to the same path remains finite.  This does
not force a change of geometry.  Under a bounded Bures kinetic budget, the
intrinsic consequence is instead \(\dot\nu\to0\); a suitable decelerating
parameterization can reach the finite-distance boundary with finite action.

The radial potential \(\Phi\) sharpens this result.  It diverges at a boundary
that is nevertheless at finite Bures distance, and its Bures gradient generates
uniform covariance dilation.  On fixed-Hamiltonian thermal families it equals
dimensionless Helmholtz free energy.  Outside those families it remains a
geometric entropy-scaling defect and should not be promoted to a thermodynamic
state function without additional structure.

The analysis also places a firm limit on measurement interpretations.  A metric
can identify stiff and soft directions; it cannot supply an apparatus,
conditioning rule, stochastic record, or transverse force.  Gaussian
measurement models can remain on the Gaussian boundary, while non-Gaussian
models can leave it.  A bona fide transition out of the Gaussian manifold must
be established through higher-cumulant dynamics or another explicit order
parameter, not inferred from the covariance floor alone.

At the opposite end, the symplectic correction is asymptotically suppressed.
The proved leading metric is the covariance Fisher component.  A concrete
dynamics intrinsic to the Bures--Wasserstein geometry is already available:
its gradient flow arises as the overdamped Rayleigh reduction of a Hamiltonian
covariance lift \cite{KerskensBWHamiltonian}.  What remains open in the present
comparison is not the existence of Bures--Wasserstein dynamics, but the
derivation of a mechanism that selects it over the quantum Bures or
Fisher--Rao branches.  The Hamiltonian lift by itself neither fixes the
relative action-matching parameter \(\eta\) nor implies the conditional
minimum-log-determinant rule.  Likewise, a finite upper endpoint requires an
independent scale and belongs to an enlarged thermodynamic theory rather than
to the bare Bures manifold.

Several mathematical questions remain.  The first is the extension of the
radial-potential identity beyond fixed Williamson frames.  The second is
whether Bures geodesics
between Gaussian states remain Gaussian; unlike Gaussian \(W_2\) interpolation,
this closure is not assumed here.  Finally, a concrete higher-cumulant stability
calculation is required to determine when a particular dynamics undergoes a
non-Gaussian bifurcation.

\subsection{Outlook: Higher cumulants and non-Gaussian dynamics}
\label{sec:higher-cumulants-outlook}

The boundary and determinant comparisons developed here are confined to the
manifold of centered Gaussian states, on which the covariance matrix provides
a complete state description.  Outside this manifold, higher connected
moments
\begin{equation}
    K=\bigl(\kappa^{(3)},\kappa^{(4)},\ldots\bigr)
\end{equation}
supply additional coordinates that are invisible to covariance geometry.

One may formally enlarge the state description from \(\Sigma\) to
\((\Sigma,K)\) and study perturbations transverse to the Gaussian sector
\(K=0\).  Such an extension is kinematic until a concrete evolution law is
specified.  In particular, neither the degeneracy of the radial covariance
cometric nor the conditional determinant-branch crossover generates a
non-Gaussian trajectory by itself.

A physical transverse evolution could arise from, for example, a Lindblad
generator with a nonquadratic Hamiltonian or nonlinear jump operators, a
non-Gaussian measurement and conditioning rule, or a specified interacting
many-body model.  Linearization about the Gaussian manifold would then have
the schematic form
\begin{equation}
    \frac{d}{dt}\delta K
    =\mathcal L_{\perp}(\Sigma,\eta)\,\delta K
    +J_{\perp}(\Sigma,\eta),
    \label{eq:transverse-linearization}
\end{equation}
where \(\mathcal L_{\perp}\) governs transverse stability and \(J_{\perp}\)
represents direct non-Gaussian forcing.  If \(J_{\perp}=0\), a local loss of
Gaussian stability would require
\begin{equation}
    \max\operatorname{Re}\operatorname{spec}
    \mathcal L_{\perp}(\Sigma_c,\eta)=0.
    \label{eq:transverse-instability}
\end{equation}
Neither \(\mathcal L_{\perp}\) nor \(J_{\perp}\) is derived in the present
work.

No result established here implies that the transverse-instability condition
coincides with the triple-volume point
\begin{equation}
    (x,\eta)=(\sqrt{\varphi},\sqrt{\varphi}).
\end{equation}
The latter is a kinematic junction of three covariance-level determinant
densities.  Whether a particular microscopic model places a non-Gaussian
stability threshold at, near, or independently of this junction is a
falsifiable open question.  Answering it requires deriving
\(\mathcal L_{\perp}\) from an explicit generator and comparing its stability
locus with the conditional Gaussian branch diagram.

\section{Conclusion}

The Gaussian Bures manifold is globally organized by the uncertainty floor and
the relative scale of its symplectic correction.  Its full-rank interior ends
at the Williamson boundary, where the radial primal metric diverges and the
dual form develops null directions.  This boundary remains at finite Bures
distance and retains a finite intrinsic pure-Gaussian geometry tangentially.
On fixed-Williamson-frame radial submanifolds, the entropy-scaling defect is an
exact Bures potential for covariance dilation and diverges at the same edge.
None of these facts alone forces a non-Gaussian transition.

In the opposite limit, \(\nu/\hbar\to\infty\), the symplectic correction is
suppressed relative to the Fisher block, producing the intrinsic asymptotic
classical limit of the geometry.  The Williamson floor has a different status:
it is a rank-changing quantum boundary with residual pure-Gaussian geometry,
not a second classical limit.  Classical statistics can be extracted there
only after a specified measurement channel, commuting observable algebra, or
conditioning rule maps the quantum family to classical outcome distributions.
The geometry constrains the regular and null directions available to that
operational restriction, while the measurement determines which directions
are statistically accessible.  It does not derive the channel or its dynamics.
This contrast between intrinsic asymptotic classicality and measurement-induced
boundary statistics is the principal global feature of the Gaussian Bures
geometry.

Separately, the one-mode determinant comparison has a codimension-two junction
at \((x,\eta)=(\sqrt{\varphi},\sqrt{\varphi})\).  Under the conditional
minimum-log-determinant rule this junction separates direct Bures--BW branch
switching from a sequential Bures--Fisher--Rao--BW diagram.  The equality is
kinematic.  A formal cotangent Gaussian integral interprets the rule, but no
common three-channel dynamics is derived; the junction is therefore neither a
tensorial Bures-to-BW limit nor a thermodynamic phase transition.

\appendix

\section{Derivation of the single-mode QFI}
\label{app:qfi}

Because Eq.~\eqref{eq:lyapunov} is covariant under fixed symplectic changes of
basis, the calculation can be made at \(\varphi=0\).  Let
\begin{equation}
    D=\begin{pmatrix}e^{2r}&0\\0&e^{-2r}\end{pmatrix},
    \qquad \Sigma=\nu D.
\end{equation}
For \(\partial_\nu\Sigma=D\),
\begin{equation}
    \mathfrak G_\nu=\frac{D^{-1}}{\nu^2-\frac14}
\end{equation}
solves Eq.~\eqref{eq:lyapunov}, and hence
\begin{equation}
    F_{\nu\nu}=\frac12\Tr(\mathfrak G_\nu D)
    =\frac{1}{\nu^2-\frac14}.
\end{equation}

For squeezing,
\begin{equation}
    \partial_r\Sigma
    =2\nu\begin{pmatrix}e^{2r}&0\\0&-e^{-2r}\end{pmatrix}.
\end{equation}
Substitution of a diagonal symmetric gain into Eq.~\eqref{eq:lyapunov} gives
\begin{equation}
    F_{rr}=\frac{4\nu^2}{\nu^2+\frac14}.
\end{equation}

For phase-space rotation,
\begin{equation}
    \partial_\varphi\Sigma
    =2\nu\sinh(2r)
    \begin{pmatrix}0&1\\1&0\end{pmatrix}
\end{equation}
up to the sign convention for \(R\).  The off-diagonal gain yields
\begin{equation}
    F_{\varphi\varphi}
    =\frac{4\nu^2\sinh^2(2r)}{\nu^2+\frac14}.
\end{equation}
The trace inner products between radial, diagonal-traceless, and off-diagonal
sectors vanish, so all cross terms are zero.

\section{Boundary and coordinate conventions}
\label{app:conventions}

The manuscript uses \(\Sigma_{\rm vac}=\Id_2/2\) and
\(g^{\rm B}=F/4\).  Other drafts in this programme use a symplectic eigenvalue
with vacuum floor \(1\).  Under \(\nu_{(1)}=2\nu_{(1/2)}\), coordinate metric
coefficients transform by pullback; individual numerical factors must not be
compared without this change of variables.

At \(\nu=1/2\), the state loses rank.  The singular limit of the full-rank SLD
QFI is therefore not identical to the metric intrinsic to the pure-state
stratum.  Every boundary statement in the main text specifies which of these is
being used \cite{Safranek2017boundary,Safranek2018}.

\section*{Code availability}
The symbolic and numerical verification script
\texttt{verify\_bures\_geometry.py} and its output log (53 deterministic
checks, SymPy/NumPy) are archived at Zenodo, DOI:~[https://doi.org/10.5281/zenodo.21497092].  The
checks cover the single-mode QFI, the covariance-cometric spectra, the
multimode class weights and \(m^2\) nullity with its \(u(m)\) identification,
the auxiliary positivity cone, the determinant identities and golden-ratio
equality with exact slopes, the conditional branch ordering, and the
finite-rate radial cost separation.
\section*{Acknowledgments}
The author acknowledges the use of AI assistants for structural brainstorming, language refinement, and \LaTeX{} editing during the drafting of this manuscript. The author bears full responsibility for the accuracy and originality of the scientific arguments and equations presented here.
\bibliographystyle{apsrev4-2}
\bibliography{biblio}

\end{document}